\documentclass[acmsmall,nonacm]{acmart}

\usepackage[utf8]{inputenc}
\usepackage{multirow}
\usepackage{amsthm}
\usepackage{algorithm}
\usepackage{algorithmic}
\usepackage{siunitx}
\usepackage[english]{babel}
\usepackage{csquotes}
\usepackage{xspace} %
\usepackage{paralist} %
\usepackage{mathtools} %
\usepackage{subcaption}
\usepackage{pgfplots}
    \pgfplotsset{compat=newest,unit code/.code={\si{#1}},plot coordinates/math parser=false,grid style={lightgray}}
    \usepgfplotslibrary{units,external,groupplots}
    \usetikzlibrary{positioning,angles,quotes,patterns,shapes,calc,backgrounds,fit}
    \tikzsetexternalprefix{tikz/}
    \tikzstyle{block} = [draw, rectangle, minimum height=2em, minimum width=5em]
\tikzstyle{addon} = [draw, rectangle, rounded corners]
\tikzstyle{pinstyle} = [pin edge={<-,thin,black}]
\tikzstyle{pinstyle2} = [pin edge={->,thin,black}]
\tikzstyle{mult} = [draw, isosceles triangle]
\tikzstyle{circ} = [draw, circle]
\tikzstyle{coord} = [coordinate]
\tikzstyle{circ2} = [draw, circle,minimum width=3pt, inner sep=0]
\tikzset{>=latex}
\tikzset{radiation/.style={{decorate,decoration={expanding
waves,angle=90,segment length=4pt}}}}
\usepackage{pifont}
    \newcommand{\cmark}{\textcolor{green!50!black}{\ding{51}}}
    \newcommand{\xmark}{\textcolor{red}{\ding{55}}}
\usepackage{ifthen}
    \newboolean{authnotes}
    \setboolean{authnotes}{false}

\graphicspath{{./img/}}

\newtheorem{theo}{Theorem}
\newtheorem{lem}{Lemma}
\newtheorem{defi}{Definition}

\newtheorem{cor}{Corollary}
\newtheorem{remark}{Remark}

\DeclareSIUnit{\belmilliwatt}{Bm}
\DeclareSIUnit{\dBm}{\deci\belmilliwatt}

\DeclareMathOperator*{\E}{\mathbb{E}}
\DeclareMathOperator*{\Var}{Var}

\DeclareMathOperator*{\Tr}{Tr}
\newcommand{\norm}[1]{\left\lVert#1\right\rVert}
\newcommand{\mah}[1]{\left\lVert#1\right\rVert_\mathrm{M}}
\newcommand{\abs}[1]{\left\lvert#1\right\rvert}
\newcommand*\diff{\mathop{}\!\mathrm{d}}
\DeclareMathOperator*{\R}{\mathbb{R}}

\let\originalleft\left
\let\originalright\right
\renewcommand{\left}{\mathopen{}\mathclose\bgroup\originalleft}
\renewcommand{\right}{\aftergroup\egroup\originalright}

\newcommand\figref[1]{Fig.~\ref{#1}}
\newcommand\tabref[1]{Table~\ref{#1}}
\newcommand\secref[1]{Sec.~\ref{#1}}

\newcommand{\etal}{et~al.\xspace}
\newcommand{\eg}{e.g.,\xspace}
\newcommand{\ie}{i.e.,\xspace}

\newcommand{\capt}[1]{\mdseries{\emph{#1}}}
\newcommand{\cf}{cf.\xspace}
\newcommand{\iid}{i.i.d.\xspace}
\newcommand{\cps}{CPS\xspace}

\newcommand{\ap}{AP\xspace}
\newcommand{\lti}{LTI\xspace}

\newcommand{\cp}{CP\xspace}

\newcommand{\dpp}{DPP\xspace}

\newcommand{\nAgent}{\ensuremath{N}\xspace}
\newcommand{\nMsg}{\ensuremath{M}\xspace}
\newcommand{\nMsgCntrl}{\ensuremath{M_\mathrm{C}}\xspace}
\newcommand{\nMsgApp}{\ensuremath{M_\mathrm{A}}\xspace}
\newcommand{\horizon}{\ensuremath{H}\xspace}
\newcommand{\aggsize}{\ensuremath{W_P}\xspace}

\ifthenelse{\boolean{authnotes}}
{

\newcommand{\fm}[1]{\footnote{{\bf\color{blue} Fabian: #1}}}
\newcommand{\mz}[1]{\footnote{{\bf\color{blue} Marco: #1}}}
\newcommand{\db}[1]{\footnote{{\bf\color{green!50!black} Dominik: #1}}}
\newcommand{\st}[1]{\footnote{{\bf\color{purple!90!black} Sebastian: #1}}}
\newcommand{\ch}[1]{\footnote{{\bf\color{orange!50!black} Carsten: #1}}}
\newcommand{\ag}[1]{\footnote{{\bf\color{purple!50!black} Alexander: #1}}}
}
{

\newcommand{\fm}[1]{}
\newcommand{\mz}[1]{}
\newcommand{\db}[1]{}
\newcommand{\st}[1]{}
\newcommand{\ch}[1]{}
\newcommand{\ag}[1]{}
}

\setcopyright{acmlicensed}
\acmJournal{TCPS}
\acmYear{2021}
\acmVolume{1}
\acmNumber{1}
\acmArticle{1}
\acmMonth{1}
\acmPrice{15.00}
\acmDOI{10.1145/3502299}

\acmSubmissionID{TCPS-2021-0018.R1}

\begin{document}

\author{Fabian Mager}
\orcid{0000-0003-0468-0691}
\authornote{Both authors contributed equally to this work.}
\affiliation{%
	\department{Networked Embedded Systems Lab}
	\institution{TU Dresden}
	\streetaddress{Helmholtzstra{\ss}e 18}
	\postcode{01069}
	\city{Dresden}
	\country{Germany}
}
\email{fabian.mager@tu-dresden.de}

\author{Dominik Baumann}
\orcid{0000-0001-7340-2180}
\authornotemark[1]
\affiliation{%
	\department{Data Science in Mechanical Engineering}
	\institution{RWTH Aachen University}
	\streetaddress{Dennewartstra{\ss}e 27}
	\postcode{52068}
	\city{Aachen}
	\country{Germany}
}
\email{dominik.baumann@dsme.rwth-aachen.de}

\author{Carsten Herrmann}
\orcid{0000-0002-2804-318X}
\affiliation{%
	\department{Networked Embedded Systems Lab}
	\institution{TU Dresden}
	\streetaddress{Helmholtzstra{\ss}e 18}
	\postcode{01069}
	\city{Dresden}
	\country{Germany}
}
\email{carsten.herrmann@tu-dresden.de}

\author{Sebastian Trimpe}
\orcid{0000-0002-2785-2487}
\affiliation{%
	\department{Data Science in Mechanical Engineering}
	\institution{RWTH Aachen University}
	\streetaddress{Dennewartstra{\ss}e 27}
	\postcode{52068}
	\city{Aachen}
	\country{Germany}
}
\email{trimpe@dsme.rwth-aachen.de}

\author{Marco Zimmerling}
\orcid{0000-0003-1450-2506}
\affiliation{%
	\department{Networked Embedded Systems Lab}
	\institution{TU Dresden}
	\streetaddress{Helmholtzstra{\ss}e 18}
	\postcode{01069}
	\city{Dresden}
	\country{Germany}
}
\email{marco.zimmerling@tu-dresden.de}

\renewcommand{\shortauthors}{F.~Mager et al.}

\authorsaddresses{%
Authors’ addresses: Fabian Mager, Carsten Herrmann, Marco Zimmerling, Networked Embedded Systems Lab, Center for Advancing Electronics Dresden (cfaed), TU Dresden, Helmholtzstraße 18, 01069 Dresden, Germany, \{fabian.mager, carsten.herrmann, marco.zimmerling\}@tu-dresden.de;
Dominik Baumann, Sebastian Trimpe, Data Science in Mechanical Engineering, RWTH Aachen University, Dennewartstraße 27, 52068 Aachen, Germany, \{dominik.baumann, trimpe\}@dsme.rwth-aachen.de.}

\title{Scaling Beyond Bandwidth Limitations: Wireless Control With Stability Guarantees Under Overload}

\begin{abstract}
An important class of cyber-physical systems relies on multiple agents that jointly perform a task by coordinating their actions over a wireless network.
Examples include self-driving cars in intelligent transportation and production robots in smart manufacturing.
However, the scalability of existing control-over-wireless solutions is limited as they cannot resolve overload situations in which the communication demand exceeds the available bandwidth.
This paper presents a novel co-design of distributed control and wireless communication that overcomes this limitation by dynamically allocating the available bandwidth to agents with the greatest need to communicate.
Experiments on a real cyber-physical testbed with 20 agents, each consisting of a low-power wireless embedded device and a cart-pole system, demonstrate that our solution achieves significantly better control performance under overload than the state of the art.
We further prove that our co-design guarantees closed-loop stability for physical systems with stochastic linear time-invariant dynamics.
\end{abstract}



\maketitle

\section{Introduction}
\label{sec:intro}

Distributed control over wireless networks is essential for cyber-physical systems (\cps) in which multiple agents work on a common task.
Examples include mobile robots jointly manufacturing a product~\cite{baumann2020manufacturing,wang2016implementing} and drones flying in formation in a rescue mission~\cite{Hayat2016}.
To support emerging multi-agent \cps, a tight integration and co-design of wireless \mbox{communication and control is needed that:}
\begin{itemize}
	\item \emph{Facilitates distributed control.} To coordinate their activities, each agent must be capable of exchanging messages with every other agent.
	In this way, each agent can drive a local control loop based on local sensor readings (\eg a drone can stabilize its flight), while in addition communication with other agents allows to solve a distributed control task (\eg drone swarm keeping a desired formation).
	This is commonly referred to as \emph{multi-agent systems}~\cite{Lunze2014multi}.
	\item \emph{Tames and accounts for network imperfections.} Control of dynamical systems like drone swarms requires information exchange every few hundred milliseconds across large distances~\cite{Preiss2017}.
	Thus, multi-hop communication with bounded latency and high reliability is crucial.
	Moreover, because wireless communication is notoriously unreliable, occasional message losses and communication delays must be accounted for by the control design.
	\item \emph{Caters for small cost, weight, form factor, and energy consumption.} Depending on the application scenario it can be beneficial, if not necessary, to deploy the entire multi-agent \cps on low-cost, low-power embedded hardware with small weight and form factor, for example, to not exceed the maximum payload of a drone or to support remote energy-harvesting sensors~\cite{baumann2020manufacturing}.
\end{itemize}

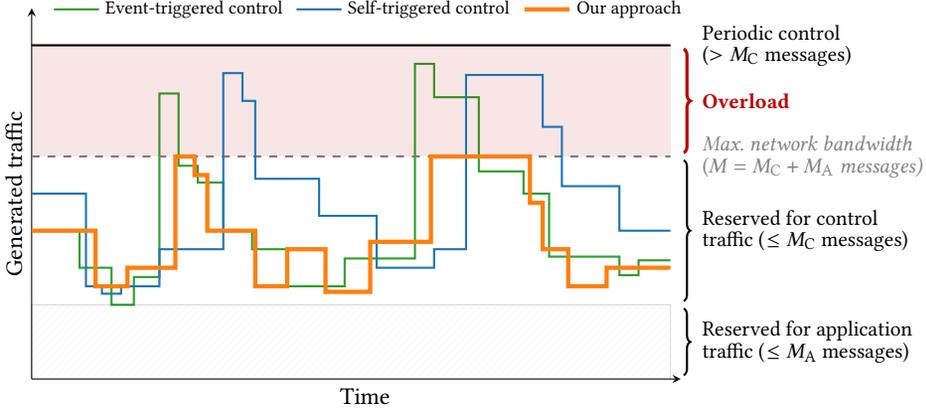
\begin{figure*}[!tb]
	\centering
	\definecolor{pred_color}{RGB}{255, 127, 14}
\definecolor{periodic_color}{RGB}{0, 0, 0}
\definecolor{etc_color}{RGB}{44, 160, 44}
\definecolor{stc_color}{RGB}{31, 119, 180}
\definecolor{red_color}{RGB}{189, 0, 0}

\begin{tikzpicture}
\begin{axis}[
 domain=0:100,
  width=\textwidth,
 height=0.32\textheight,
axis lines = left,
xlabel={Time},
xmin=0, xmax=145,
ylabel={Generated traffic},
label style={font=\small},
ymin=10, ymax=110,
 ticks = none,
 x axis line style={shorten >=105, shorten <=0},
 x label style={at={(axis description cs:0.36,0)},anchor=north},
 legend style={font=\scriptsize, at={(axis cs:2,115)},anchor=north west, draw=none, /tikz/every even column/.append style={column sep=0.1cm}},
 legend columns=3,
 legend cell align={left},
]

\addplot[thick, etc_color]
table[row sep = crcr]{
0       50\\
7.5      50\\
7.5      40\\
12.5    40\\
12.5    30\\
16      30\\
16      37.5\\
20      37.5\\
20      87\\
23      87\\
23      67.5\\
26      67.5\\
26      63\\
30      63\\
30      50\\
34.5    50\\
34.5    45\\
40      45\\
40      35\\
49      35\\
49      42.5\\
60      42.5\\
60      95\\
63      95\\
63      86\\
70      86\\
70      66\\
77      66\\
77      60\\
81      60\\
81      43\\
88      43\\
92      43\\
92      38\\
95      38\\
95      42\\
100     42\\
};
\addlegendentry{Event-triggered control}

\node [] (per_ctrl) at(105,115) {};

\addplot[thick, periodic_color, forget plot] {100} node[pos=1.25] (periodic_helper) {};
\node [periodic_color, anchor=west, align=left, font=\footnotesize\linespread{0.8}\selectfont] at(per_ctrl.west|-periodic_helper) (periodic_control)
{Periodic control\\($>\nMsgCntrl$ messages)};

\addplot[dashed, thick, gray, forget plot] {70}node[pos=1.25](bandwidth_helper){};
\draw[pattern=north east lines, opacity=0.15] (0,0) rectangle (100, 30);

\node [gray, anchor=west, align=left, font=\footnotesize\linespread{0.8}\selectfont] at(per_ctrl.west|-bandwidth_helper)(bandwidth)
{\textit{Max. network bandwidth}\\($\nMsg=\nMsgCntrl+\nMsgApp$ \textit{messages)}};

\draw [red_color,decorate,very thick,decoration={brace},yshift=0pt,xshift=5pt]
(100,99) -- (100,71)node [midway](overload_helper) 
{};

\node [red_color, anchor=west] at(per_ctrl.west|-overload_helper)(overload){\footnotesize{\textbf{Overload}}};
\draw[fill=red_color, draw=none, opacity=0.1] (0,70) rectangle (100, 100);

\draw [decorate,decoration={brace},thick,yshift=0pt,xshift=5pt]
(100,69) -- (100,31) node [midway] (ctrl_traffic_helper) {};

\node [black, anchor=west, align=left, font=\footnotesize\linespread{0.8}\selectfont] at(per_ctrl.west|-ctrl_traffic_helper)(ctrl_traffic)
{Reserved for control\\traffic ($\le \nMsgCntrl$ messages)};

\draw [decorate,decoration={brace},thick,yshift=0pt,xshift=5pt]
(100,29) -- (100,11) node [midway](reserved_helper) 
{};

\node [black, anchor=west, align=left, font=\footnotesize\linespread{0.8}\selectfont] at(per_ctrl.west|-reserved_helper)(ctrl_traffic)
{Reserved for application\\traffic ($\le \nMsgApp$ messages)};

\addplot[thick, stc_color]
table[row sep = crcr]{
0       60\\
8.5     60\\
8.5     35\\
11      35\\
11      33\\
14      33\\
14      35\\
20      35\\
20      45\\
30      45\\
30      92.5\\
33      92.5\\
33      85\\
35      85\\
35      64\\
45      64\\
45      54\\
54      54\\
54      40\\
63      40\\
63		45\\
68		45\\
68      92\\
80      92\\
80      78\\
83      78\\
83      62\\
92      62\\
92      50\\
100     50\\
};
\addlegendentry{Self-triggered control}

\addplot[ultra thick, pred_color]
table[row sep = crcr]{
0       50\\
10      50\\
10      35\\
15      35\\
15      40\\
22.5    40\\
22.5    70\\
25.5    70\\
25.5    65\\
27.5    65\\
27.5    50\\
35      50\\
35      35\\
40      35\\
40      45\\
46      45\\
46      33.5\\
53      33.5\\
53      47\\
62.5    47\\
62.5    70\\
78      70\\
78      57.5\\
80      57.5\\
80      45\\
84      45\\
84      35\\
90      35\\
90      40\\
100     40\\
};
\addlegendentry{Our approach}

\end{axis}

\end{tikzpicture}
	\caption{
		Illustration of the problem and approaches.
		\capt{
			The network bandwidth \nMsg available per time step (update interval) can be used to transmit up to \nMsgCntrl control messages and up to $\nMsgApp$ application messages.
			A system is overloaded if the control traffic generated by periodic control exceeds \nMsgCntrl and therefore \nMsg.
			Using event- or self-triggered control in an overloaded system cannot prevent temporary bandwidth exceedances, resulting in unpredictable behavior.
			Our approach ensures that the generated control traffic never exceeds \nMsgCntrl, and closed-loop stability can be provably guaranteed.
		}
	}
	\label{fig:overview}
\end{figure*}

While meeting these requirements is challenging in itself, multi-agent \cps also face an \emph{overload problem}. %
To illustrate, let us consider the scenario shown in \figref{fig:overview}.
The bandwidth \nMsg available per time step (update interval) can be used to transmit up to \nMsgCntrl control messages and up to \nMsgApp messages carrying other application data, such as photos and video streams~\cite{Hayat2016,baumann2020manufacturing}.
However, as applications become ever more sophisticated---requiring, for example, more agents, shorter update intervals, and higher-volume data streams---the generated traffic inevitably exceeds the available bandwidth \nMsg.
Specifically, a system is \emph{overloaded}, when the control traffic generated by periodic control exceeds \nMsgCntrl.
Periodic control in an overloaded system causes additional message loss equal to the amount of bandwidth exceedance, which may make it impossible to guarantee closed-loop stability and to achieve the required control performance.

Advances in wireless communication technology cannot solve the overload problem: The network bandwidth remains a limited resource that is ultimately outrun by increasing application demands.
On the other hand, as illustrated in \figref{fig:overview} and detailed in \secref{sec:relWork}, existing approaches such as event- and self-triggered control can only reduce the generated control traffic \emph{on average} compared to periodic control.
However, in an overloaded system they typically cannot prevent situations where the bandwidth is temporarily exceeded.
The behavior of the system during such situations (\eg in terms of closed-loop stability) is unpredictable, which is unacceptable for critical \cps applications requiring a priori guarantees~\cite{rajkumar2010cyber}.

\paragraph{Contributions}
We present the design, analysis, and real-world evaluation of a wireless \cps that addresses the overload problem, while meeting all of the above-mentioned requirements.
Using our approach, the generated control traffic never exceeds the fraction \nMsgCntrl of the bandwidth reserved for control (see \figref{fig:overview}), and we derive stability guarantees for the entire multi-agent CPS.

As described in Secs.~\ref{sec:overview} to~\ref{sec:integrationStability}, our approach is based on a novel co-design and tight integration of wireless communication and control.
The key idea is to determine how urgent each agent needs to transmit control data, and to assign the available control bandwidth \nMsgCntrl in every update interval to those agents that currently have the highest need.
Although the communication system we design is highly reliable, occasional message loss cannot be avoided due to the limited time for communication.
Our control design accounts for such message loss as well as communication delays.
By tightly integrating communication and control, we reduce the jitter caused by imperfect synchronization of distributed hardware components in real \cps to the point where it can be neglected.
As a result, our overall solution is amenable to a formal end-to-end analysis of all relevant \cps components (communication, control, and physical system), which allows us to prove closed-loop stability for heterogeneous agents with stochastic linear time-invariant (\lti) dynamics.

We evaluate our approach on a 20-agent \cps testbed.
Each agent consists of a low-power wireless embedded device and a cart-pole system, whose dynamics are representative of mechanical systems found in real-world applications~\cite{Astrom2008,trimpeCSM12}.
The 20 agents form a 3-hop network, exchanging control traffic every \SI{100}{\milli\second} to synchronize the movement of their carts.
Our experiments demonstrate that, in the scenarios we tested, the overall \cps is stable as predicted by our theoretical analysis despite external disturbance.
The experimental results further show that our approach synchronizes the carts better than a highly optimized periodic baseline, while using fewer control messages.

In summary, this work makes the following contributions:
\begin{itemize}
 \item We present the first practical wireless \cps design that addresses the overload problem. With this, we improve the scalability of \cps toward future applications with increasing demands. %
 \item We formally prove that our wireless \cps design guarantees closed-loop stability for heterogeneous agents (\ie physical systems) with potentially different stochastic LTI dynamics.
 \item Real-world experiments on a 20-agent CPS %
 testbed confirm our theoretical results and demonstrate an improved control performance while using fewer control messages.
\end{itemize}

\section{Problem and Related Work}
\label{sec:relWork}

This section defines the research problem we tackle in this paper and reviews relevant prior work.

\subsection{Problem Formulation}

\paragraph{Scenario}
Motivated by emerging applications in search and rescue, manufacturing, or construction~\cite{Hayat2016,wang2016implementing,baumann2020manufacturing}, we consider wireless \cps consisting of \nAgent heterogeneous agents that jointly work on a distributed control task.
Each agent runs a controller that computes actuator commands based on local sensor readings and information received from other agents.
While local readings allow each agent to, for instance, stabilize itself, communication is essential to solve the distributed task, such as flying in formation.
To this end, the agents are equipped with radio frequency (RF) transceivers to exchange messages over a wireless multi-hop network.

The agents' physical dynamics and the required control performance govern the update interval at which control information is to be exchanged in a many-to-many fashion among the agents.
In this work, we target distributed control of mechanical systems requiring update intervals on the order of tens to hundreds of milliseconds~\cite{Akerberg2011,Preiss2017}.
Conversely, the bandwidth of the wireless network determines the number \nMsg of messages that can be exchanged within each update interval.
Out of these, as illustrated in \figref{fig:overview}, only $\nMsgCntrl \! < \! \nMsg$ messages can carry control information.
This is because the wireless network is also used to transmit other application data (\eg video streams, photos, status and configuration data), which occupy $\nMsgApp = \nMsg - \nMsgCntrl$ messages per update interval.

\paragraph{Overload Problem}
The state of the art with respect to the outlined application scenario supports at most $\nAgent = 5$ agents at an update interval of \SI{50}{\milli\second} when no application traffic is transmitted~\cite{baumann2019control,mager2019feedback}.
This is insufficient for many envisioned \cps applications requiring tens to hundreds of agents, ever shorter update intervals to realize more sophisticated control tasks, and the continuous collection of high-volume data streams, for example, to feed machine-learning models~\cite{Hayat2016,wang2016implementing,baumann2020manufacturing}.

The bottleneck is the limited network bandwidth \nMsg.
While advances in wireless technology can increase \nMsg, the required infrastructure costs may not be economically viable~\cite{frankston21consumer}.
Moreover, the traffic volumes of machine-to-machine communication to enable monitoring and control are expected to see annual growth rates of up to \SI{50}{\percent} over the next ten years~\cite{statista}, quickly outrunning any increase in~\nMsg.
We refer to a system as overloaded when the control traffic generated by periodic control exceeds the available control bandwidth \nMsgCntrl.
As a result, it becomes impossible to guarantee stability and achieve the desired control performance with periodic control methods.

\subsection{Related Work}

How to achieve high-performance control under limited communication resources has been widely studied.
However, as discussed below, most prior approaches cannot solve the overload problem.
A few theoretical control concepts can in principle address the problem, but none of these works considers the challenges of integrating control with a real network, neither wired nor wireless.

\paragraph{Event- and Self-Triggered Control}
Event-triggered control (ETC) and self-triggered control (STC) methods aim to efficiently use the limited communication bandwidth~\cite{heemels2012introduction,miskowicz2018event}.
To this end, they only let agents transmit control information when needed (\eg some error exceeds a threshold) instead of letting all agents transmit control information in every update interval as in standard periodic control.
However, while ETC and STC can reduce the control traffic \emph{on average} compared to periodic control, they cannot solve the overload problem: At any point in time, it can happen that more agents signal communication needs than the network can support, as illustrated in \figref{fig:overview}.
How to resolve such situations and provide stability guarantees under overload is an unsolved problem.

Further, using ETC, agents make communication decisions instantaneously, which leaves the communication system no time to reallocate unused bandwidth to other agents, wasting precious resources.
Using STC, an agent decides about the next time it needs to communicate at the current communication instant, so unused bandwidth can be reallocated.
However, there is no way to react between two communication instants; that is, the agent cannot react to unforeseen disturbances, which negatively affects control performance and stability.

\paragraph{Predictive Triggering}
Predictive triggering can handle disturbances by letting agents decide at every time step if they need to communicate some time in the future~\cite{trimpe2016predictive,trimpe2019resource}.
Moreover, Mastrangelo \etal extend predictive triggering toward non-binary communication decisions, where a priority measure based on the probability of exceeding a threshold is used to schedule communication~\cite{mastrangelo2019predictive}.
Unlike the binary communication decisions in ETC and STC, this approach can in principle address the overload problem, which is why we adopt it.
However, compared with all prior work on predictive triggering, we \emph{(i)} propose an improved priority measure that is efficiently computable on resource-constrained hardware, \emph{(ii)} address the challenges of integrating predictive triggering with a real wireless communication system, \emph{(iii)} conduct a formal stability analysis, and \emph{(iv)} validate our overall co-design on a real-world \cps testbed.

\paragraph{Contention Resolution}
Although it has been shown~\cite{mastrangelo2019predictive} that predictive triggering yields better performance than contention resolution~\cite{molin2011optimal,mamduhi2017error}, we discuss it here as an alternative theoretical concept that can in principle address the overload problem.
Besides the inability of some contention resolution algorithms to support heterogeneous agents~\cite{mamduhi2017error}, which is a common requirement in practice, none of the existing algorithms (see, \eg~\cite{ramesh2016performance,balaghi2018decentralized,mamduhi2017error,demirel2018deepcas,molin2011optimal}) has been integrated with a real network.
Instead, the algorithms are exclusively evaluated in simulation, making assumptions about a potential communication system that are not backed up through real-world experiments.

\paragraph{Online Scheduling}
Recently, a few distributed~\cite{Zhang2021,Modekurthy2019HART} and autonomous~\cite{Modekurthy2019} scheduling approaches for wireless control systems have been proposed.
The goal of these approaches is to adjust sampling periods and communication schedules in response to unexpected external disturbances and varying wireless link qualities.
Although their distributed operation resembles our scheduling approach and also shares the goal of adapting to external disturbances, the scheduling criterion and techniques are fundamentally different.
For instance, rather than adjusting the schedules to link quality changes, such changes are effectively accounted for by our synchronous transmission and network coding based communication system, thereby hiding them from the scheduler.

\begin{table}[!tb]
\caption{
    Comparison to prior practical co-designs of control and wireless communication that predict communication demands and have been validated using experiments on real physical systems and wireless networks.
    \capt{
        With fast physical systems we refer to the ability of a co-design to support update intervals of a few hundreds of milliseconds, which is typically required to control mechanical systems, such as a quadcopter swarm~\cite{Preiss2017}.
        Co-designs that can quickly react to disturbances are here defined as those that can spontaneously react based on current measurements in contrast to, \eg STC designs, which decide about their next communication instant at the current one and cannot react in-between.
    }
}
\label{tab:rel_work}
\begin{tabular}{cccccc}
\toprule
\multirow{2}{*}{Work} & Fast physical & Multi-hop & Quickly react to & Stability & Addresses \\
& systems & networks & disturbances & guarantees & overload \\
\midrule
Araujo \etal \cite{Araujo2014}  & \xmark & \xmark & \xmark & \cmark & \xmark \\
Saifullah \etal \cite{Saifullah2014} & \xmark & \cmark & \cmark & \xmark & \xmark \\
Santos \etal \cite{santos2015aperiodic}  & \cmark & \xmark & \xmark & \xmark & \xmark \\
Baumann \etal\cite{baumann2019control} & \cmark & \cmark & \xmark & \xmark & \xmark \\
\textbf{This work} & \cmark & \cmark & \cmark & \cmark & \cmark \\
\bottomrule
\end{tabular}
\end{table}

\paragraph{Practical Control-Communication Co-Designs}
\tabref{tab:rel_work} qualitatively compares our and prior practical co-designs that predict communication demands and validate the integration of control with wireless communication against the dynamics of real physical systems and real wireless networks.

Araujo \etal use STC to control a quadruple tank process over a single-hop network with update intervals of a few seconds~\cite{Araujo2014}.
Saifullah \etal present a multi-hop solution for power management in data centers, using update intervals of 20 seconds or longer~\cite{Saifullah2014}.
While their control design is not explicitly based on STC, it exhibits similar properties, including the inability to cope with overload.
The same holds for the co-designs by Santos \etal~\cite{santos2015aperiodic} and Baumann \etal~\cite{baumann2019control}.
Both employ an STC approach and demonstrate control of fast physical systems; however, only the solution by Baumann \etal supports control over multi-hop networks.
None of the works based on STC can spontaneously react to disturbances, and only~\cite{Araujo2014} provides stability guarantees.

We also note a few other recent control-communication co-designs that use STC~\cite{ma2018efficient,Ma2020} or ETC~\cite{bhatia2021control,trobinger2021wireless} to reduce communication and that have been evaluated on real wireless networks.
In contrast to our work, these co-designs target slow physical systems (\eg water distribution networks~\cite{bhatia2021control,trobinger2021wireless}) requiring update intervals on the order of seconds and have only been evaluated on simulated physical systems.
Most importantly, none of them addresses the overload problem.

In summary, the co-design proposed in this paper is the first to address the overload problem, while providing several other properties (see \tabref{tab:rel_work}) essential for emerging \cps applications.

\section{Overview of Co-Design Approach}
\label{sec:overview}

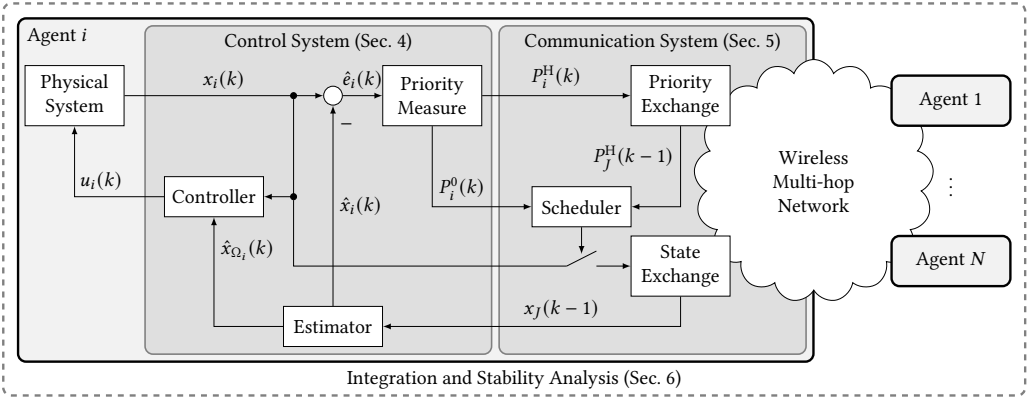
\begin{figure*}[!tb]
	\centering
	\newcommand*{\jname}{\ensuremath{J}}
	\newcommand*{\icomp}{\ensuremath{\bar I}}
	\scalebox{0.75}{%
\begin{tikzpicture}
	
\begin{scope}[every node/.append style={rectangle, draw=gray, thick, rounded corners, fill=lightgray!50!white, minimum height = 15.5em, text depth = 15em}]
	\node[minimum width=17.5em] (ctrl) {Control System (\secref{sec:control})};
	\node[minimum width=15.5em, right=0.3em of ctrl] (comm) {Communication System (\secref{sec:communication})};
\end{scope}

\node[rectangle, minimum height=2em, minimum width=5em, below right = 2.5em and 1em of ctrl.north west](sys_i){};
\node[block, fill=white,  minimum height=3em, align=center, left = 2em of sys_i](phys_sys){Physical\\System};
\node[block, fill=white, below = 3em of sys_i](ctrl_i){Controller};
\node[circ, fill=white, right = 3em of sys_i](err_sum){};
\node[block, fill=white, below = 10em of err_sum](pred_i){Estimator};
\node[block, fill=white, minimum height=3em, right=2em of err_sum, align=center](viol_prob){Priority\\ Measure};
\coordinate[right = 7.5em of viol_prob](prio_exc_west);
\node[block, fill=white, below right=3em and 2.5em of viol_prob](sched){Scheduler};
\coordinate[below=8.5em of prio_exc_west](state_exc_west);
\node[coord]at($(prio_exc_west)!0.5!(state_exc_west)$)(exc){};
\node[draw, fill=white, right=3.3em of exc, anchor=west, cloud, cloud puffs=17, cloud puff arc=140, aspect=1.0, inner ysep=2.3em, align=center](nw){Wireless\\ Multi-hop\\ Network};
\node[block, fill=white, minimum height=3em, anchor=west, align=center] at (prio_exc_west) (prio_exc){Priority\\ Exchange};
\node[block, fill=white, minimum height=3em, anchor=west, align=center] at (state_exc_west) (state_exc){State\\ Exchange};

\begin{scope}[every node/.append style={draw,very thick,rounded corners,fill=lightgray!20!white, minimum height=2.5em, minimum width=6em}]
	\node at ($(nw)+(30:8em)$) (agent1) {Agent $1$};
	\node at ($(nw)+(-30:8em)$) (agentN) {Agent \nAgent};
\end{scope}
\node at ($(agent1)!0.5!(agentN)$) {$\vdots$};

\coordinate (center) at ($(ctrl.east)!0.5!(comm.west)$);

\node[circ2, fill=black]at($(sys_i.east)!0.5!(err_sum.west)$)(branch_xi){};
\node[circ2, fill=black]at(branch_xi|-ctrl_i)(branch_xi_ctrl){};

\node[coord]at(sched|-state_exc)(switch_center){};

\draw[->] (phys_sys) -- node[midway, above]{$x_i(k)$} (err_sum);
\draw[->] (ctrl_i) -| node[pos=0.35, above]{$u_i(k)$} (phys_sys);
\draw[->] (branch_xi) |- (ctrl_i.east);
\draw[->] (pred_i) -| node[pos=0.85, right]{$\hat{x}_{\Omega_i}(k)$} (ctrl_i);
\draw[->] (pred_i) -- node[midway, right]{$\hat{x}_i(k)$} node[pos=0.9, right]{$-$} (err_sum);
\draw[->] (err_sum) -- node[midway, above]{$\hat{e}_i(k)$} (viol_prob);
\draw[->] (viol_prob) -- node[midway, above]{$P_i^\mathrm{\horizon}(k)$} (prio_exc);
\draw[->] (prio_exc) |- node[pos=0.2,left]{$P_{\jname}^\mathrm{\horizon}(k-1)$} (sched.east);
\draw[->] (viol_prob) |- node[pos=0.65,above]{$P_i^0(k)$} (sched);
\draw (branch_xi_ctrl) |- ($(switch_center)+(-0.75em,0em)$);
\draw[->] ($(switch_center)+(0.75em,0em)$) -- (state_exc);
\draw ($(switch_center)+(-0.75em,0em)$) -- +(1.5em, 0.75em);
\draw[->] (sched.south) -- ($(switch_center)+(0em,0.6em)$);
\draw[->] (state_exc.south) |- node[pos=0.7, above]{$x_\jname(k-1)$} (pred_i);

\begin{scope}[on background layer]
	\node[rectangle,draw,very thick,rounded corners, fill=lightgray!20!white, align=left, fit=(ctrl) (comm) (phys_sys)] (agenti) {};
	\node[anchor=north west, inner sep=0.5em] at (agenti.north west) {Agent $i$};
\end{scope}

\node[fit=(agent1) (agenti) (agentN)] (int_core) {};
\node[anchor = north, inner sep=0, outer sep=0] at (int_core.south) (int_text) {Integration and Stability Analysis (\secref{sec:integrationStability})};
\node[rectangle,draw=gray,very thick,rounded corners,dashed, fit=(int_core) (int_text)] {};

\end{tikzpicture}
}
	\caption{
		System architecture of our proposed co-design approach.
		\capt{
			The CPS consists of \nAgent heterogeneous agents working on a common task.
			Each agent controls a physical system.
			To coordinate their actions, all agents communicate through a wireless multi-hop network.
			$\Omega_i$ denotes the set of agents relevant for the control task,
			and \jname denotes all agents whose
			information has been successfully exchanged in the previous time step.
		}
	}
	\label{fig:predTrigFramework}
\end{figure*}

To overcome the overload problem, we propose a novel co-design approach that integrates control and communication at \emph{design time} (\ie prior to operation) and at \emph{run time} (\ie during operation).
The design of the communication system tames network imperfections to the extent possible, and the design of the control system takes the remaining imperfections in terms of message loss and delay into account.
During operation, the control system continuously reasons about the future need for communication of each of the \nAgent agents.
The communication system time-synchronizes the agents and adapts to the communication needs by dynamically assigning the \nMsgCntrl messages available for control traffic in each update interval to the agents with the highest need.
As illustrated in \figref{fig:overview}, our co-design approach ensures that the generated control traffic never exceeds \nMsgCntrl, and closed-loop stability can be provably guaranteed.

\figref{fig:predTrigFramework} shows the overall system architecture and the data flow between the building blocks inside each agent.
The novel control system we design consists of three building blocks:
\begin{enumerate}
	\item[\ref{block:controller}] The \emph{controller} performs the desired control task for the physical system connected to agent $i$.
	It locally stabilizes the individual system, while at the same time it exploits many-to-all communication capabilities to realize the distributed control task.
	\item[\ref{block:measure}] The \emph{priority measure} indicates an agent's future need for communication to meet the required control performance.
	The measure makes it possible to compare communication needs independent of the specific control tasks and for different physical dynamics of heterogeneous~agents.
	\item[\ref{block:estimator}] The \emph{estimator} provides the controller and the priority measure with estimates of the current state of all agents based on the available information received over the network.
\end{enumerate}

We jointly design the control system with a novel wireless communication system, which consists of the following three building blocks:
\begin{enumerate}
	\item[\ref{block:stateExchange}] The \emph{state exchange} distributes, in each update interval $k$, \nMsgCntrl control messages over the wireless multi-hop network to all \nAgent agents.
	The set of at most \nMsgCntrl agents that are allowed to transmit their messages may change from one update interval to the next.
	Because of unavoidable packet losses when communicating over wireless channels, it is possible that an agent receives only a subset of the transmitted messages; however, our experiments and empirical evidence from prior work~\cite{herrmann2018mixer} show that the probability is extremely low (\ie less than \SI{0.01}{\percent}).
	\item[\ref{block:priorityExchange}] The \emph{priority exchange} builds on the mechanism from \ref{block:stateExchange} to distribute each agent's current priority measure among all \nAgent agents in the system.
	\item[\ref{block:scheduler}] The distributed \emph{scheduler} dynamically assigns the \nMsgCntrl control messages available in each update interval $k$ to the agents with the highest priorities.
	Every agent computes the global communication schedule locally.
	If an agent does not have full information about the latest priorities, it is guaranteed that it does not disturb the communication of the other agents.
\end{enumerate}

\paragraph{Road Map} In Secs.~\ref{sec:control} and \ref{sec:communication}, we detail the co-design of the control and wireless communication systems along the different building blocks \ref{block:controller} to \ref{block:scheduler}.
Then, we describe the careful integration of control and communication, and formally prove stability of the entire \cps for physical systems with stochastic \lti dynamics in \secref{sec:integrationStability}.
Finally, \secref{sec:eval} complements our analytical results through real-world experiments on a \cps testbed.

\section{Predictive Triggering and Control~System}
\label{sec:control}

In this section, we present the design of the control system.
We first provide a model of the physical system and then show how the different building blocks of \figref{fig:predTrigFramework} are designed.
In particular, we introduce the controller (\ref{block:controller}) and the estimator (\ref{block:estimator}) in \secref{sec:ctrl_arch}, and derive the priority measure (\ref{block:measure}) in \secref{sec:viol_prob}.

\subsection{Control System Model}
\label{sec:sys_model}

We consider a collection of $\nAgent$ stochastic LTI systems.
The dynamics of the physical system $i$ (\cf \figref{fig:predTrigFramework}) of this collection are given by
\begin{align}
\label{eqn:gen_sys_lin_dyn}
x_i(k+1) = A_ix_i(k) + B_iu_i(k) + v_i(k),
\end{align}
with discrete time index $k\in\mathbb{N}$, state $x_i(k)\in\R^n$, input $u_i(k)\in\R^m$, $A_i$ and $B_i$ matrices of appropriate dimensions, and process noise $v_i(k)\in\R^n$, which we assume to follow a normal distribution with zero mean and variance $\Sigma_\mathrm{v_i}$.
Further, we assume the noise processes of individual agents to be uncorrelated, \ie $\E[v_i(k)v_j(k)]=0$ for all $i\neq j$, and $k$.
The time needed by the communication system to distribute up to $\nMsg$ messages determines the update interval, which represents the fixed length of a discrete time step $k$.

We consider a distributed control problem, \ie the input $u_i(k)$ of agent $i$ at time $k$ depends not only on its own state but also on the state of (possibly all) other agents.
Assuming static linear feedback, the control law for agent $i$ becomes
\begin{align}
\label{eqn:gen_ctrl_law}
u_i(k) = F_{ii}x_i(k) + \sum_{j\in\Omega_i}F_{ij}\hat{x}_{ij}(k),
\end{align}
with feedback matrices $F_{ii},F_{ij}\in\R^{m\times n}$, $\Omega_i$ the set of all agents whose state is relevant to agent $i$ (that is, those for which $F_{ij}$ is non-zero), and $\hat{x}_{ij}(k)$ agent $i$'s estimate of agent $j$'s state.
The estimate of other agents' states is based on the information transmitted over the wireless network (\cf \figref{fig:predTrigFramework}).
However, we account for the case that only $\nMsgCntrl<\nAgent$ agents can transmit information in every update interval $k$, \ie an agent $i$ may not receive state updates from all other agents $j\in\Omega_i$ at time step $k$.

Due to the clearer presentation, we assume throughout the paper that all agents operate with the same update interval $k$.
Nevertheless, our design also supports heterogeneous agents with different update intervals as long as the communication system can provide the required update interval.
Intuitively, agents that are stable at longer update intervals, \ie agents with slower dynamics, can also be stabilized at shorter update intervals.
We evaluate different types of systems in \secref{sec:eval}, where the heterogeneity is reflected in different communication needs.

\subsection{Control Architecture}
\label{sec:ctrl_arch}

Based on the model~\eqref{eqn:gen_sys_lin_dyn}, we now introduce the estimator (\ref{block:estimator}) that lets agent $i$ estimate agent $j$'s current state.
Afterward, we discuss
the feedback matrices $F_{ii}$ and $F_{ij}$ of the controller (\ref{block:controller}).

\paragraph{State Estimation}
As stated in the previous section, agent $i$ needs information of the other agents' states to compute its control input.
However, only $\nMsgCntrl<\nAgent$ agents can transmit their state information over the communication network.
Thus, for all agents from which agent $i$ did not receive an update, it needs to estimate the current state based on the data it received so far.
Further, information sent over the network is delayed by one time step and subject to message loss.
An estimate of agent $j$'s current state, compensating sporadic message exchange, transmission delays, and message loss, can be obtained via
\begin{align}
\label{eqn:est_state_j}
\begin{split}
\hat{x}_{ij}(k) = \begin{cases}
(A_j+B_jF_{jj})x_j(k-1) + B_j\sum_{\ell\in\Omega_j}F_{j\ell}\hat{x}_{i\ell}(k-1) &\text{ if } \kappa_j(k)=1\land\phi_j(k)=1\\
(A_j+B_jF_{jj})\hat{x}_{ij}(k-1) + B_j\sum_{\ell\in\Omega_j}F_{j\ell}\hat{x}_{i\ell}(k-1) &\text{ otherwise,}%
\end{cases}
\end{split}
\end{align}
where $\kappa_j(k)$ denotes, whether ($\kappa_j(k)=1$) or not ($\kappa_j(k)=0$) agent $j$ transmitted its state in the current round.
Possible message loss is captured by $\phi_j(k)$ ($\phi_j(k)=0$ if the message was lost and $\phi_j(k)=1$ otherwise).
The intuition behind~\eqref{eqn:est_state_j} is as follows: in case of successful communication, agent $i$ receives agent $j$'s state.
To compensate for the transmission delay, agent $i$ makes a one step ahead prediction.
In case of no communication (either intentionally or due to message loss), it propagates its last estimate of agent $j$'s state ($\hat{x}_{ij}$).

\paragraph{Control Objective}
While we consider distributed control, the individual agents may be unstable, \ie the $A_i$ in~\eqref{eqn:gen_sys_lin_dyn} may have eigenvalues with absolute value greater than one.
Thus, each agent needs to locally stabilize itself \emph{and} solve the distributed control task.
To make the controller design precise, we consider synchronization as an example of distributed control.
We here understand synchronization as trying to have the states of all agents evolve as close as possible.
That is, we want to keep the error $e_{ij}(k)\coloneqq x_i(k)-x_j(k)$ between the states of any two agents $i$ and $j$ small.
Alternatively, it is also possible to only synchronize a subset of the states.
Synchronizing the states (or parts thereof) of multiple agents is a frequently considered problem setting in distributed control and also known under the terms consensus or coordination~\cite{lunze2012synchronization}.

\paragraph{Controller Design}
For ease of presentation, we assume a two-agent setting in the following. However, as we also show in the experimental evaluation, the design straightforwardly extends to more agents.
Based on the control objective, we start by formulating a cost function
\begin{align}
	\label{eqn:lqr_cost}
	\begin{split}
	J = \lim_{K\to\infty}\frac{1}{K}\E & \left[\sum_{k=0}^{K-1}\sum_{i=1}^2\left(x_i(k)^\transp Q_ix_i(k) + u_i(k)^\transp R_i u_i(k)\right)\right.\\
	&\left.+(x_1(k)-x_2(k))^\transp Q_\mathrm{sync}(x_1(k)-x_2(k))\vphantom{\sum_{k=0}^{K-1}}\right],
	\end{split}
\end{align}
with positive semidefinite matrices $Q_i$ and $Q_\mathrm{sync}$, which penalize deviations from the equilibrium state and from the synchronization objective, and positive definite matrix $R_i$, which penalizes high control inputs.
Since both agents have an estimate of the other agent's current state, obtained from~\eqref{eqn:est_state_j}, we can rewrite the term in the summations in~\eqref{eqn:lqr_cost}, using the augmented state $\tilde{x}(k) = [x_1(k),x_2(k)]^\transp$ and augmented input $\tilde{u}(k) = [u_1(k),u_2(k)]^\transp$, as
\begin{align}
\label{eqn:lqr_cost_augmented}
\begin{split}
&\tilde{x}^\mathrm{T}(k)
\begin{pmatrix}
Q_1+Q_\text{sync}&-Q_\text{sync}\\
-Q_\text{sync}&Q_2+Q_\text{sync}
\end{pmatrix}
\tilde{x}(k)
 +\tilde{u}^\mathrm{T}(k)
 \begin{pmatrix}
R_1&0\\
0&R_2
\end{pmatrix}
 \tilde{u}(k).
 \end{split}
\end{align}
We now seek to find the optimal feedback controller that minimizes the cost function $J$.
Since the cost function $J$ puts an emphasis on the stability of each individual agent (through $Q_i$) and on the synchronization objective (through $Q_\mathrm{sync}$), solving the optimization problem results in finding a trade-off between both objectives.
This trade-off can be influenced through the choice of $Q_i$ and $Q_\mathrm{sync}$.
Using the formulation in~\eqref{eqn:lqr_cost_augmented}, this problem can be solved using standard tools from linear optimal control~\cite{Anderson2007} and yields a static linear feedback controller as in~\eqref{eqn:gen_ctrl_law}.
Note that this implies that the agents do not need to solve the optimization problem at run time.
Instead, the optimal controller can be computed offline before the system operation commences.

\subsection{Priority Measure}
\label{sec:viol_prob}

We now discuss how to derive a priority measure (\ref{block:measure})  that enables the communication system to allocate the \nMsgCntrl available control messages to the agents with the highest needs.
If state estimates~\eqref{eqn:est_state_j} were perfect, there would be no need for any communication.
However, since we consider noisy dynamics (\cf~\eqref{eqn:gen_sys_lin_dyn}), the estimates will, over time, start to deviate from the true state.
Thus, we seek to transmit information as soon as the error between estimated and true state becomes too large.
Therefore, agent $i$ computes the same\footnote{Estimates may diverge in case of message loss. We discard impact of message loss for the triggering design and analyze its impact in the stability analysis in \secref{sec:stab_analysis}.} estimate other agents have of its state and derives the estimation error $\hat{e}_i(k)\coloneqq x_i(k) - \hat{x}_{ii}(k)$.
Ideally, we would now like agent $i$ to transmit its state whenever $\hat{e}_i(k)$ exceeds some threshold.
However, we need to \emph{(i)} announce communication needs in advance, \ie we cannot use the current $\hat{e}_i(k)$ to decide about communication, and \emph{(ii)} we need a measure that quantifies agent $i$'s communication demand in a generic way, instead of a binary decision as it is the case with ETC and STC.

We first introduce the measure that we use to quantify whether the estimation error is ``too large.''
This measure is the squared Mahalanobis distance~\cite{mahalanobis1936generalised} of the estimation error from the origin,
\begin{align}
\label{eqn:mahalanobis_distance}
\mah{\hat{e}_i(k)}^2 \coloneqq \E[(\hat{e}_i(k))]^\transp\Var[\hat{e}_i(k)]^{-1}\E[(\hat{e}_i(k))].
\end{align}
The distance between an observation and a random variable is in one dimensional settings often measured by assessing how many standard deviations it is away from the mean.
The Mahalanobis distance extends this idea to multiple dimensions.
For $\Var[\hat{e}_i(k)] = I_n$, with $I_n$ the $n\times n$ identity matrix, the Mahalanobis distance reduces to the Euclidean distance.

Directly using the squared Mahalanobis distance to schedule agents would yield a viable triggering strategy for homogeneous agents.
Yet, for heterogeneous agents, the estimation errors may be in different orders of magnitude and, thus, not comparable.
Further, instead of an instantaneous decision, we need to announce communication needs in advance.
As shown in \figref{fig:predTrigFramework}, at time step $k$, each agent receives states that were sent at $k-1$, \ie there is a delay of one time step.
Thus, we consider the probability of the squared Mahalanobis distance exceeding a predefined threshold $\delta_i$ in $\horizon+1$ time steps.
The parameter $\horizon$ denotes how many time steps the communication system needs to reschedule resources.
To compute the Mahalanobis distance, we need the expected value and the variance of the error $\hat{e}_i(k+\horizon+1)$.
Since the error is the difference between current state and estimated state, its expected value and variance can be calculated using expected value and variance of the state at time $k+\horizon+1$.
Assuming no communication between $k$ and $k+\horizon+1$ and exploiting that the noise sequences are uncorrelated, those are given by
\begin{subequations}
\begin{align}
\E[x_i(k+\mathrm{\horizon}+1)] &= A_i^{\mathrm{\horizon}+1}x_i(k) + \sum_{s=0}^\mathrm{\horizon} A_i^sB_iu_i(k-1-s)\\
\label{eqn:mah_dist_var}
\Var[x_i(k+\mathrm{\horizon}+1)] &= \sum_{s=0}^\mathrm{\horizon} A_i^s\Var[v_i(k-1-s)](A_i^s)^\transp,
\end{align}
\end{subequations}
where $\Var[v_i(k)]=\Sigma_\mathrm{v_i}$ for all $k$.
Note that~\eqref{eqn:mah_dist_var} is constant, \ie it does not depend on current data.
Thus, the inverse needed in~\eqref{eqn:mahalanobis_distance} can be computed \emph{a priori}, leaving only matrix multiplications that need to be done at run time.

We define our priority measure as
\begin{align}
\label{eqn:viol_prob_analytic}
P_i^\mathrm{\horizon}(k) = \frac{\gamma\left(\frac{n}{2},\frac{\delta_i-\mah{\hat{e}_i(k+\mathrm{\horizon}+1)}^2}{2}\right)}{\Gamma\left(\frac{n}{2}\right)},
\end{align}
with $\gamma(\cdot)$ the lower incomplete gamma function, $\Gamma(\cdot)$ the gamma function, and $n$ the dimensionality of the physical system.
This priority measure is the closed-form expression of the probability that the squared Mahalanobis distance of a normally distributed random variable with zero mean exceeds $\delta_i-\mah{\hat{e}_i(k+\mathrm{\horizon}+1)}^2$~\cite{gallego2013mahalanobis}.
Technically, it is therefore the probability that the Mahalanobis distance of the estimation error will either exceed $\delta_i$ or shrink by more than $\delta_i-\mah{\hat{e}_i(k+\mathrm{\horizon}+1)}^2$ and not exactly the sought-after probability.
However, this does not change the qualitative result (systems with higher probability of exceeding $\delta_i$ will have a higher priority measure) and, in that way, allows to rank systems by priority.
Further, we do not need the exact probability of a system exceeding the threshold for any of our technical results.
Note that~\eqref{eqn:viol_prob_analytic} can only be evaluated if the estimation error is smaller than the threshold.
If the threshold is already exceeded, we can reverse the arguments of $\gamma(\cdot)$ to compute the probability that the estimation error falls again below the threshold.

\section{Adaptive Communication System}
\label{sec:communication}
The second part of our \cps co-design is the communication system.
As outlined in \secref{sec:overview}, we consider the problem that the communication bandwidth is not sufficient to meet the demands of application and periodic control.
We have \nMsgCntrl control messages for \nAgent agents with $\nAgent > \nMsgCntrl$, and a potentially high number of application messages \nMsgApp. %
To quickly and reliably exchange all $\nMsg = \nMsgApp + \nMsgCntrl$ messages (\ref{block:stateExchange}), we design a wireless protocol based on the communication primitive Mixer~\cite{herrmann2018mixer}.
\secref{sec:scalability} details how we increase the scalability of Mixer to support larger multi-agent systems, whereas \secref{sec:triggering_support} describes novel extensions to support distributed control based on our predictive triggering approach.
Specifically, we conceive a scheme for the efficient and reliable exchange of priorities among agents (\ref{block:priorityExchange}), which we then leverage for our distributed scheduling design (\ref{block:scheduler}) to adaptively allocate messages to agents with the highest communication needs.

\begin{figure}[!tb]
	\centering
	\includegraphics[width=\linewidth]{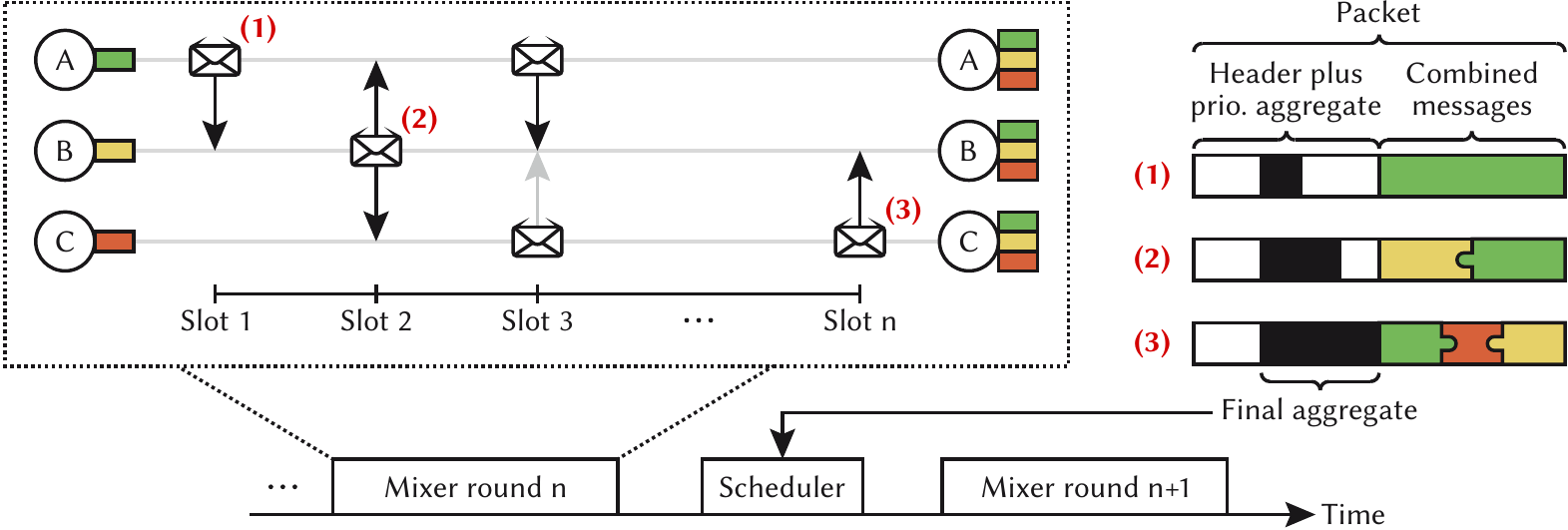}
	\caption{
		Overview of the wireless communication protocol.
		\capt{
			Mixer realizes the many-to-all message exchange in a network (colors represent messages).
			A Mixer round consists of synchronized time slots, where in each slot nodes (A, B, C) decide whether and what to transmit.
			Nodes build packets by randomly combining already received information, \eg node B sends a combination of yellow and green in slot 2.
			Priorities are exchanged and aggregated as part of Mixer's header.
			At the end of the round all nodes have all messages.
			The scheduler runs locally on each node between rounds and uses the final aggregate to determine the global message allocation for the next round.
	}}
	\label{fig:comProtocol}
\end{figure}

\subsection{Communication Support for Scalable Multi-agent Systems}
\label{sec:scalability}
In prior approaches capable of distributed control over wireless multi-hop networks~\cite{mager2019feedback,baumann2019fast,baumann2019control} the latency to exchange \nMsg messages scales as $O(\nMsg T)$, with $T$ the time required for the dissemination of a single message.
With this scaling behavior the number of messages that can be exchanged in a given update interval is rather small.
For instance, the experiments in \cite{baumann2019fast} support a maximum of five agents with a \SI{50}{\milli\second} update interval an no application traffic.
To significantly improve on this number, we adopt and enhance the recently proposed Mixer many-to-all communication primitive~\cite{herrmann2018mixer}, which will form the basis of our wireless protocol.
Mixer approaches the order-optimal scaling of $O(\nMsg + T)$, resulting in a significant improvement over the $O(\nMsg T)$ scaling for larger \nMsg and $T$.

Below we give a brief overview of how Mixer works and describe how we improve its performance by
\emph{(i)} porting it from IEEE~802.15.4 to the faster Bluetooth Low Energy (BLE) physical layer to reduce $T$, and
\emph{(ii)} making use of the most recent network state information to communicate more effectively.
\secref{sec:triggering_support} details how we enhance Mixer's functionality for a systematic co-design with the predictive triggering based control system.

\paragraph{Mixer Protocol Operation}
Mixer is a fast, efficient, and reliable communication primitive that disseminates a set of \nMsg messages among all nodes in a network by using synchronous transmissions~\cite{zimmerling2021} and random linear network coding~\cite{ho2006}.
As shown in \figref{fig:comProtocol}, a Mixer \emph{round} consists of synchronized time \emph{slots}, and in each slot the nodes decide independently whether to transmit or listen for a \emph{packet}.
Mixer packets consist of a protocol header and a combination of 1 to \nMsg messages as \emph{payload}.
Such a combination is created by choosing a random subset of already received packets and combining their payloads via XOR.
This coding-based approach is the main reason for the improved scaling behavior, as it makes better use of the communication channel.
The slotted communication continues until all \nMsg messages ($\nMsg = 3$ for the example in \figref{fig:comProtocol}) can be decoded by all nodes, or the predefined maximum duration of a round has elapsed.

Synchronous transmissions exploit the capture effect~\cite{Ferrari2011} which describes that, under certain conditions, multiple transmissions can overlap (collide) at a receiver, but one of the packets is likely to be received successfully.
This is comparable to a conversation between two people while at the same time other people are talking from a greater distance.
The voice signals interfere in the air, however, it is likely possible to understand the person nearby.
For example, in slot 3 in \figref{fig:comProtocol} node B successfully receives the packet from node A despite the concurrent transmission of node C as a result of the capture effect.

\begin{figure}[!tb]
	\centering
	\input{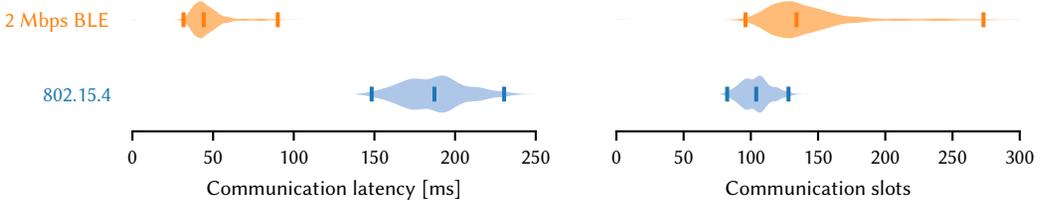}
	\caption{
		Distribution of communication latency using Mixer with the 802.15.4 and \SI{2}{Mbps} BLE physical layer.
		\capt{
			Although communication using BLE requires more slots, because of the lower robustness of this physical layer, the higher data rate reduces the length of each slot and thus the resulting latency by \num{4.2}$\times$ (median).
			The strokes mark 1st, 50th (median), and 99th percentiles.
	}}
	\label{fig:bleVs154}
\end{figure}

\paragraph{Mixer Using Bluetooth Low Energy}
The original Mixer implementation targets IEEE 802.15.4-compliant low-power RF transceivers with a limited data rate of \SI{250}{kbps}.
Another physical layer (PHY) in the low-power wireless domain is Bluetooth Low Energy (BLE), which supports data rates of up to \SI{2}{Mbps} in the Bluetooth 5 standard, which is \num{8}$\times$ faster than 802.15.4.
To benefit from the faster PHY, we implemented Mixer using BLE on the Nordic Semiconductor nRF52840 platform.

This platform also supports the 802.15.4 PHY, which makes it easier to compare the end-to-end performance of both variants.
For this purpose, we run experiments involving several thousand Mixer rounds on a 31-node testbed at Graz University of Technology~\cite{schuss17competition}.
The nodes were deployed across several rooms and hallways in an office building.
With a transmit power of \SI{8}{dBm}, the nodes form networks with a diameter of 2 hops (802.15.4) and 3 hops (BLE).
Nodes exchange messages with a size of \SI{18}{bytes} in an all-to-all fashion, \ie in every round all nodes start with one message ($\nMsg = 31$) that should be disseminated to all other nodes.

The left plot in \figref{fig:bleVs154} shows the latency distribution of both Mixer variants.
We see that our BLE implementation of Mixer outperforms the original version in terms of communication latency by \num{2.5}--\num{4.7}$\times$ with a median of \num{4.2}$\times$.
The main reason why the theoretical speed-up of \num{8}$\times$ is unattainable in practice are differences in the robustness of the two physical layers:
while the 802.15.4 PHY features a spreading code, the \SI{2}{Mbps} BLE mode is uncoded and thus more susceptible to concurrent transmissions induced interference.
Consequently, the number of slots needed with BLE to exchange \nMsg messages increases compared to 802.15.4, as visible in the right plot of \figref{fig:bleVs154}.
At the same time, the higher data rate reduces the length of each slot, which leads to the overall latency improvements.

\begin{figure}[!tb]
	\centering
	\input{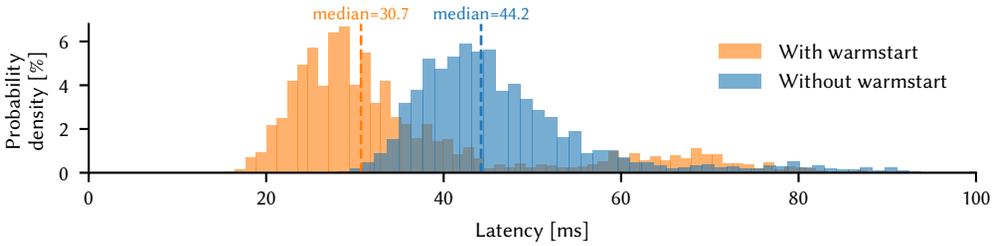}
	\caption{
		Histogram of Mixer's latency using the \SI{2}{Mbps} BLE PHY with and without warmstart.
		\capt{
			Warmstarting nodes helps them communicate more effectively at the beginning of a round and reduces latency by \SI{31}{\percent} (median).
	}}
	\label{fig:warmstart}
\end{figure}

\paragraph{Warmstart}
Wireless control with fast physical systems requires short update intervals.
For instance, drones need to exchange their positions every \SI{100}{ms} to fly in formation or to avoid crashes~\cite{chung16}.
In general and from a network perspective, the topology changes continuously in such dynamic scenarios.
However, viewed from one Mixer round to the next, the topology changes are relatively small due to the short update interval.
We leverage this observation to further improve Mixer's performance without sacrificing communication reliability.

Mixer rounds are independent of each other.
Each node has no information about its neighborhood at the beginning of a round.
This changes during a round with each packet received, allowing better and better transmit decisions to be made, and thus the message exchange to be finished more quickly.
Instead of acquiring this information from scratch in each round, we extend Mixer to allow neighborhood information to carry over from one round to the next.
Such \emph{warmstarted} nodes can communicate more effectively right from the beginning of a round.
Experiments with the same testbed setup as before show that warmstarting nodes reduces latency by 12--\SI{40}{\percent}, with a median of \SI{31}{\percent} (see \figref{fig:warmstart}). %

\subsection{Support for Predictive Triggering}
\label{sec:triggering_support}

Having established an efficient and reliable communication support for wireless control, we now turn to the core co-design questions from the perspective of the communication system:
\begin{enumerate}
	\item[\emph{(i)}] How to efficiently exchange the priorities computed by each agent within a Mixer round?
	\item[\emph{(ii)}] How to schedule communication based on the priorities such that the network bandwidth is used in the most effective way by allocating messages to the agents with the highest need?
	\item[\emph{(iii)}] How to communicate efficiently when facing varying communication demands?
\end{enumerate}
We propose novel Mixer protocol functionality to answer these three questions, as discussed next.

\paragraph{Efficient Exchange of Priorities}
With respect to their storage format the priorities $P_i^\mathrm{\horizon}(k)$ are generally much smaller in size than the full state information $x_i(k)$.
The key idea is to exploit this size advantage and use a part of the available network bandwidth for the exchange of all priorities in each Mixer round, to enable using the remaining network bandwidth in the most effective way.

The problem with this is that Mixer's coding operations require that all messages exchanged in a round must have the same size, thus, transmitting priorities and state information as separate messages would completely eliminate the size advantage of priorities.
Instead, we extend Mixer's packet header and reserve space for a priority aggregate, as shown in \figref{fig:comProtocol}.
At the beginning of a round, each agent knows only its own priority.
During a round, agents aggregate priorities from all received packets and embed the current aggregate into the packets they transmit.
Since the aggregate is part of every packet, it distributes independently of the \nMsg individual messages and propagates quickly and reliably through the network.

To minimize the overhead introduced by the priority aggregate, a compact representation is needed.
Importantly, the aggregate must contain information about the \nMsgCntrl agents with the highest priority and whether all agents contributed to the aggregate, as required by the distributed scheduling algorithm described below.
We propose two representations that satisfy these requirements.
The first option is to concatenate all $N$ priorities in the aggregate.
With \aggsize being the number of bits used for each priority value, the required number of bytes $S$ for the aggregate of \nAgent priority values is given by
\begin{subequations}
	\label{eqn:aggregate_size}
	\begin{equation}
		S = \lceil \nAgent \aggsize / 8 \rceil.
	\end{equation}
Alternatively, if $\nAgent \gg \nMsgCntrl$, we achieve a more compact representation using
	\begin{equation}
		S = \lceil (\nMsgCntrl \aggsize + \nMsgCntrl \lceil \log_2 \nAgent \rceil + \nAgent) / 8 \rceil,
	\end{equation}
\end{subequations}
which includes the \nMsgCntrl highest priorities ($\nMsgCntrl \aggsize$), together with the corresponding agent IDs ($\nMsgCntrl \lceil \log_2 \nAgent \rceil$), and one bit per node to verify that all agents have contributed ($N$).
Based on the concrete scenario parameters (\ie \aggsize, \nMsgCntrl, and \nAgent), our system automatically choses the most compact representation.

\paragraph{Distributed Scheduling}
A schedule in our communication protocol is the mapping of agents to the \nMsgCntrl available control messages, hence, it describes which agents are allowed to send a control message in the next round.
Agents derive the schedule based on the final priority aggregate to which all agents have contributed, by selecting the \nMsgCntrl agents with the highest need to send their state information.
More formally, we have that agent $i$ can send a message if its priority $P_i^\mathrm{\horizon}(k) \in \mathcal{P}_\mathrm{L}$, where $\mathcal{P}_\mathrm{L}$ denotes the set of the \nMsgCntrl highest priorities
, \emph{and} if $P_i^\mathrm{\horizon}(k)>P_\delta$ for some user-defined threshold $P_\delta$.
How we deal with equality of priorities, a result of using quantization in a real implementation, is described in \secref{sec:scenario}.

As described in \secref{sec:viol_prob}, the scheduling horizon \horizon depends on the number of time steps the communication system needs to reschedule resources.
In our system the priority aggregate always determines the schedule for the next round, so $\horizon = 1$.
In the rare event that agent $i$ does not have the final aggregate at the end of a round, it is not guaranteed that it can compute the correct global schedule and therefore the agent must not send a message.
Nevertheless, the agent participates in the communication round like all other agents that have no message assigned to them.
However, if $P_i^\mathrm{\horizon}(k) \in \mathcal{P}_\mathrm{L}$, the message assigned to agent $i$ is not used in the next round.
Since messages and priorities spread independently of each other, message loss has little to no effect on the final priority aggregate.
In our experiments, we did not record a single priority loss and only one message loss, which shows the reliability of our communication system.

\paragraph{Reducing Communication Costs}
A Mixer round ends when all nodes have received all \nMsg messages or the maximum duration of a round is reached.
Since \nMsg is a parameter set at compile time, nodes expect the same number of messages in every round, which suits periodic traffic.
If the communication demand varies from round to round, then \nMsg determines the maximum number of messages.
Since nodes do not finish the round before receiving \nMsg messages, the difference to the actual number of messages has to be filled with empty messages.
As a result, the original Mixer~\cite{herrmann2018mixer} cannot translate a lower communication demand into energy savings.%

To achieve energy savings if less than \nMsg messages need to be sent in a round, we extend Mixer so that nodes can mark unused messages instead of sending empty messages.
The idea is similar to the way the priority aggregate works.
We aggregate the information about unused messages and quickly distribute them in the network independent of individual messages.
To realize this, we can reuse a certain part of Mixer's packet header (\ie the InfoVector, see~\cite{herrmann2018mixer}) and switch each slot between its original content and the information about unused messages, thus, adding no further overhead.

\begin{figure}[!tb]
	\centering
	\input{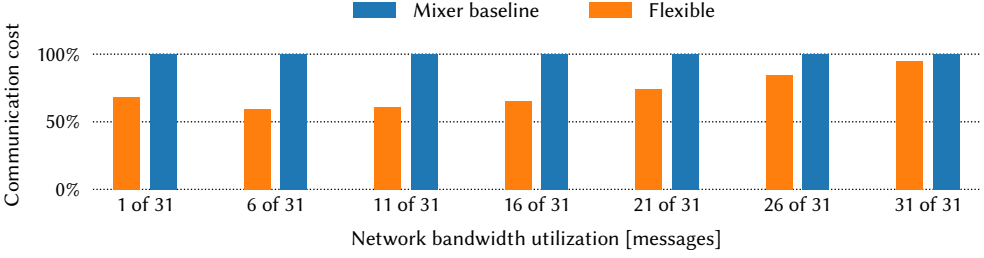}
	\caption{
		Communication costs at different levels of network bandwidth utilization.
		\capt{
			The original Mixer (baseline) cannot convert lower bandwidth utilization into energy savings.
			Enabling nodes to mark unused messages (flexible) saves up to \SI{40}{\percent} of the communication costs.
	}}
	\label{fig:weak_msgs}
\end{figure}

We evaluate the savings in terms of communication costs enabled by this modification using the experimental setup from \secref{sec:scalability}.
Now, messages also contain the priority aggregate, and we enable the warmstart feature.
We determine the communication costs by measuring the average radio-on time across all nodes during a communication round.
\figref{fig:weak_msgs} plots for the original Mixer baseline and our approach, labeled \emph{Flexible}, the resulting communication costs when different fractions of the available network bandwidth are used, from 1 used message out of 31 possible messages up to 31 used messages out of 31 possible messages.
We can see that, by marking unused messages, our Flexible approach achieves significant savings of up to \SI{40}{\percent} compared with the original Mixer baseline.
As expected, the savings increase as the network bandwidth utilization decreases.

The possibility to save energy can also be taken into account by the scheduling algorithm.
Since the communication need is announced one round in advance ($\horizon = 1$), the estimation error $\hat{e}$ can also decrease until the next time step, whereby allocated messages might no longer be needed.
For this, every agent that is allocated a message in the schedule derives $P_i^0(k)$ (as shown in \figref{fig:predTrigFramework}) by setting $\horizon = 0$ in~\eqref{eqn:mahalanobis_distance} and~\eqref{eqn:viol_prob_analytic}.
If $P_i^0(k)<P_\delta$, agent $i$ can skip the transmission and mark the message as unused to save energy.
The combined scheduling rule can be summarized as
\begin{equation}
\label{eqn:comm_dec}
\kappa_i(k+\mathrm{\horizon}+1) = 1 \iff P_i^\mathrm{\horizon}(k)\in\mathcal{P}_\mathrm{L} \land P_i^\mathrm{\horizon}(k) > P_\delta \land P_i^0(k+\mathrm{\horizon}) > P_\delta.
\end{equation}
Overall, this mechanism keeps the energy costs incurred by communication in proportion to the communication demand.

\section{Integration and Stability Analysis}
\label{sec:integrationStability}

The integration of our co-designed control and communication systems with respect to their timely interaction and execution is key to meet \cps application requirements.
We briefly characterize implementation and timing aspects of our wireless embedded system in \secref{sec:integration}, followed by a formal stability analysis of our integrated system in \secref{sec:stab_analysis}.

\subsection{Implementation Aspects}
\label{sec:integration}

\paragraph{Hardware Platform}
We leverage a dual processor platform (\dpp) that integrates an application processor (\ap) and a communication processor (\cp) together with the processor interconnect Bolt~\cite{Sutton2015a}.
The \ap executes all application tasks (sensing, actuation, control) while the \cp performs all communication tasks.
Both processors run in parallel, and the interconnect permits data exchange within formally verified timing bounds~\cite{Sutton2015a}.
\ap and \cp use ARM Cortex-M4F CPU cores running at \SI{48}{MHz} (MSP432P401R) and \SI{64}{MHz} (nRF52840), respectively.
The \cp features a low-power RF transceiver that implements the IEEE 802.15.4 and BLE physical layers.

\paragraph{Time Synchronization}
We use the start of a communication round as a global reference time, which is estimated by all \cp{}s in the network during communication.
Each \cp propagates this reference time to its respective \ap over a dedicated connection between GPIO pins.
By aligning all application tasks to the global reference time, similar to~\cite{mager2019feedback}, we ensure that the end-to-end delay and jitter among application tasks (\ie from sensing through control to actuation) is minimized~\cite{jacob2020}.
The slightly different estimates of the reference time by the \cp{}s lead to jitter, which is, however, upper bounded to the nominal length of a communication slot due to the design of Mixer.
For our BLE implementation and typical payload sizes, the jitter is on the order of a few hundred microseconds.

\subsection{Stability Analysis}
\label{sec:stab_analysis}

While the individual agents can locally stabilize themselves, they are coupled through a common control objective.
Introducing couplings can destabilize otherwise stable subsystems~\cite[Ch.~8]{lunze1992feedback}, even under periodic communication.
Moreover, we also consider larger CPS where the limited number of control messages leads to an overloaded system.
All these aspects imply that stability of the overall system does not follow straightforwardly from the stability of individual systems and, thus, needs to be analyzed.
For the stability analysis, we make the following assumptions.

\paragraph{Jitter}
The system model~\eqref{eqn:gen_sys_lin_dyn} implicitly assumes that the time step $k\to k+1$ (\ie the update interval) is constant, implying that jitter is neglected.
One time step corresponds to the time between the start of consecutive communication rounds.
The \cps scenarios described in \secref{sec:intro} use update intervals of tens to hundreds of milliseconds.
As a rule of thumb, jitter that is below \SI{10}{\percent} of the nominal update interval does not need to be compensated for~\cite[p.~48]{cervin2003integrated}.
As discussed in \secref{sec:integration}, the jitter of our \cps architecture is way below that threshold, justifying our assumption.

\paragraph{Message Loss}
For the following analysis, we assume that message loss is independent and identically distributed (\iid).
Given the extremely high reliability of Mixer with hardly any message losses during hundreds of hours of indoor and outdoor experiments~\footnote{In ~\cite{herrmann2018mixer}, this has been demonstrated with the IEEE 802.15.4 PHY. We have recorded only one message loss in about \num{50000} messages in our experiments with the BLE \SI{2}{Mbps} PHY.}~\cite{herrmann2018mixer}, the temporal correlation among those very few message losses is practically negligible.
This also justifies our second assumption: we assume that the message loss probability is below \SI{100}{\percent}, a theoretical corner case that would prohibit any stability analysis.

\paragraph{System and Control Design}
For notational convenience, we assume homogeneous systems with equal choices of $Q$ and $R$ matrices for controller design, \ie $A_i=A_j$, $B_i=B_j$, and $F_{ij}=F_{ji}$ for all $i,j$.
We, thus, drop the index $i$ for $A$ and $B$ in this section and comment on how the analysis would need to be adapted for the heterogeneous case in the end of the section.
The controller matrix is designed such that both the overall system and the individual closed-loop systems $A+BF_{ii}$ would be stable under periodic communication without message loss.
For homogeneous systems, the system with the largest Mahalanobis distance from the origin will also have the highest priority.

We define stability of the system in terms of boundedness of its second moment (\cf~\cite{ramponi2010attaining}):
\begin{defi}[Mean square boundedness]
\label{def:stability}
The system~\eqref{eqn:gen_sys_lin_dyn} is mean square bounded (MSB) if there exist $\bar{\epsilon}$ and $\rho(\epsilon)$ such that for a given $\epsilon>\bar{\epsilon}>0$, $\norm{x(0)}_2 < \rho$ implies $\sup_{k\ge 0}\E[\norm{x(k)}_2^2]\le\epsilon$.
\end{defi}
With this definition, we can state the following theorem:
\begin{theo}
\label{thm:stability}
Consider a system consisting of $\nAgent$ homogeneous agents with stochastic LTI dynamics as in~\eqref{eqn:gen_sys_lin_dyn}, a transmission channel with $0<\nMsgCntrl<\nAgent$ resources per communication round, the control law~\eqref{eqn:gen_ctrl_law}, estimation strategy~\eqref{eqn:est_state_j}, and scheduling law~\eqref{eqn:comm_dec} with $\horizon=1$.
Then, the system is MSB.
\end{theo}
The proof mainly follows the proofs in~\cite{mamduhi2017error,mamduhi2015robust}.
We here provide some intuition and comment on the main differences to~\cite{mamduhi2017error,mamduhi2015robust} while we defer the formal proof to the supplementary material.

As the controller is designed so that the overall closed-loop system is stable under reliable, periodic communication, potential instability stems from the estimation error.
Thus, we show that we can derive upper bounds for this error in a mean-squared sense, which ensures MSB of the overall system, as we show in the supplementary material.

Compared to~\cite{mamduhi2017error,mamduhi2015robust}, we assume a deterministic scheduling policy.
While in~\cite{mamduhi2017error,mamduhi2015robust} the agent with the largest error might not receive a slot due to the probabilistic scheduling policy, this cannot happen in our design.
Therefore, we can obtain tighter upper bounds on the estimation error.
Further, since we assume message loss to be \iid, we do not need to guarantee that a certain amount of messages is delivered successfully within a given time interval.

\begin{remark}
For homogeneous agents, we arrange agents according to their errors, since the system with the largest error has the largest Mahalanobis distance.
This simplifies notation since we can directly consider the errors.
However, for heterogeneous agents, where the errors may not be comparable, we can derive an upper bound on the probability of exceeding the threshold, which implies an upper bound on the error.
Thus, considering heterogeneous agents comes at no additional complexity, except for notation.
\end{remark}

\begin{remark}
Theorem~\ref{thm:stability} guarantees that the second moment is bounded but makes no statements about the actual bound.
Explicitly computing this bound would enable to assess the performance of the considered setup, which clearly depends on various parameters, such as the number of control messages $\nMsgCntrl$.
This can be used to derive adjustable performance bounds.
For a more detailed discussion, we refer the reader to~\cite{mamduhi2017error}.
How to extend the analysis to scenarios without local stabilization is discussed in Remark~\ref{rem:stab} in the supplementary material.
\end{remark}

\begin{figure}
    \centering
    \subcaptionbox{Testbed deployment (\SI{20}{\meter}~$\times$~\SI{30}{\meter}) including 6 agents with different types of real physical systems and 14 agents with simulated physical systems. \capt{The heterogeneity makes distributed control more challenging.}
    \label{fig:testbed_layout}}
    {\includegraphics[]{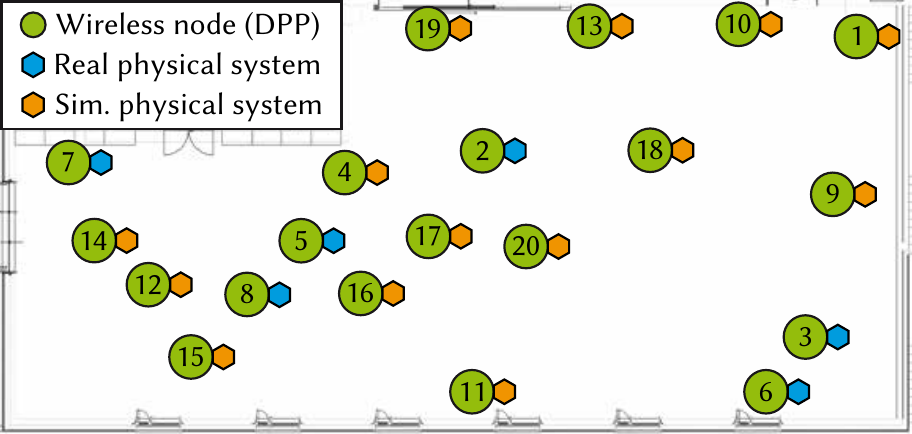}}
    \hfill
    \subcaptionbox{Cart-pole system. \capt{Controlled cart movements can stabilize the pole in an upright position.}
    \label{fig:cartpole_system}}
    {\begin{tikzpicture}[scale=0.85, every node/.style={scale=0.85}]]
\tikzstyle{ground}=[fill,pattern=north east lines,draw=none,minimum width=5,minimum height=0.1]
\tikzset{>=latex};
\draw(0.5,0.35)rectangle(-0.5,-0.35);
\draw[fill=white](-0.25,-0.5)circle(0.2);
\draw[fill=white](0.25,-0.5)circle(0.2);
\draw[ultra thick,lightgray](0,0)node(cart){}--(1,2.5);
\draw[fill=black](1,2.5)circle(0.2)node(pole){};
\draw[dashed] (0,0.5) -- (0,2.5)node(top){};
\draw[fill=black](0,0)circle(0.05);
\draw[fill=black](-0.25,-0.5)circle(0.025);
\draw[fill=black](0.25,-0.5)circle(0.025);
\node[ground,minimum width=115,anchor=north](floor)at(0,-0.7){};
\draw(floor.north east)--(floor.north west);
    \draw pic[" ",draw=black, -,  angle eccentricity=1.5,angle radius =
    2cm] {angle = pole--cart--top};
\draw[->](1.1,0)node[right]{Cart} -- (0.6,0);
\draw[->] (1.1,1.1)node[right]{Pole} -- (0.55,1.1);
\draw([shift={(0,-0.2cm)}]floor.south west) -- node[midway,draw,rectangle,inner sep = 0, minimum width=0.1em, minimum height=0.4em,label=below:0,fill=black]{}
node[pos=0,draw,rectangle,inner sep = 0, minimum width=0.1em, minimum height=0.4em,label=below:-25,fill=black]{}
node[pos=1,draw,rectangle,inner sep = 0, minimum width=0.1em, minimum height=0.4em,label=below:25,fill=black]{}
([shift={(0,-0.2cm)}]floor.south east);
\draw[->] (-1.5,-0.1)node[above]{Track} -- (-1.5,-0.6);
\node[below = 2em of floor,inner sep = 0]{Cart position $s$ [\si{\centi\meter}]};
\draw[->](-0.5,1.6)node[left, align=center]{Pole\\angle\\$\theta$ [\si{\degree}]} -- (0.3,1.6);
\end{tikzpicture}}
    \par\bigskip
    \subcaptionbox{Wireless node. \capt{Dual Processor Platform (DPP).}
    \label{fig:dpp}}
    {\includegraphics[width=0.38\linewidth, angle=270]{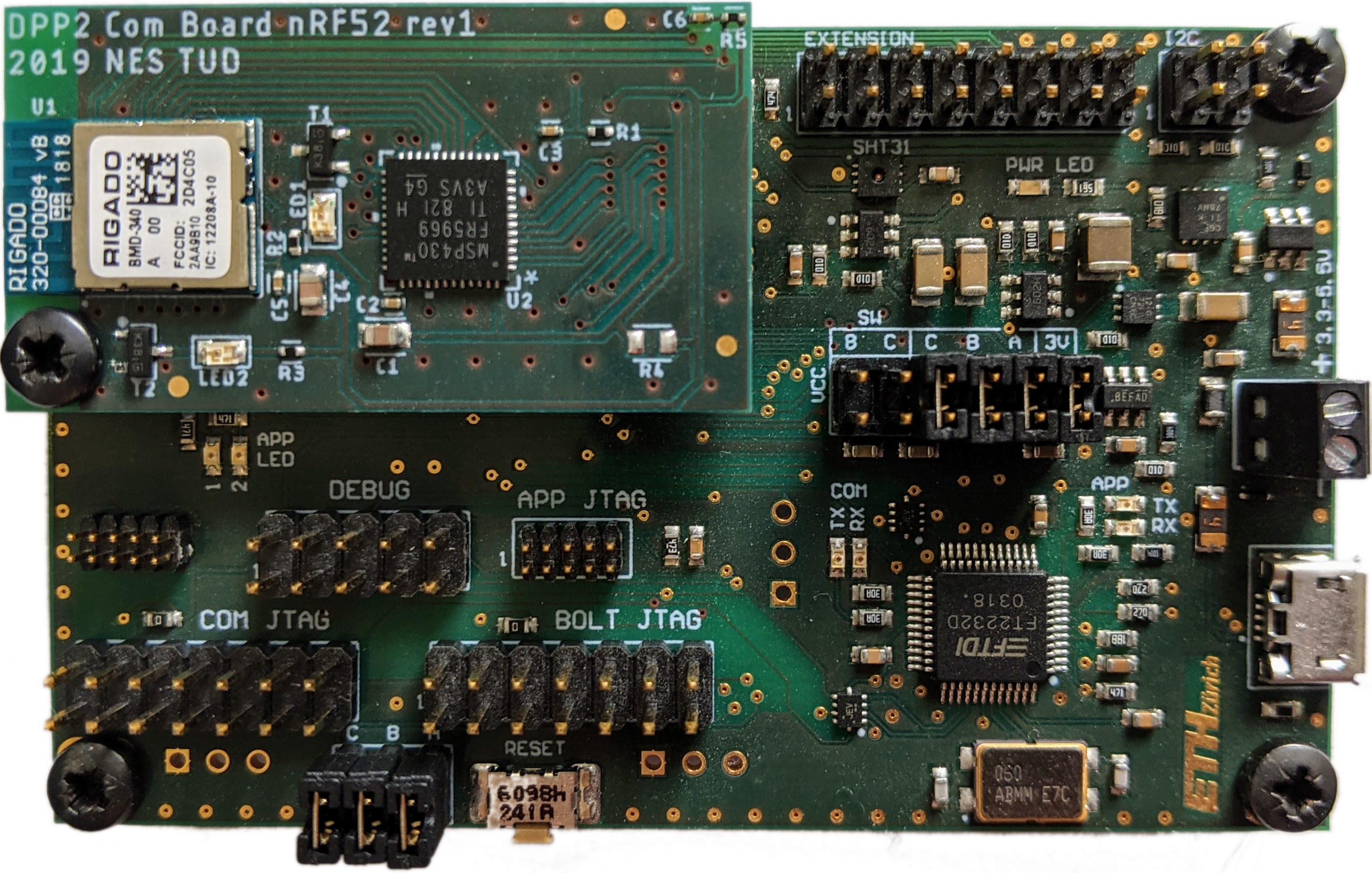}}
    \hfill
    \subcaptionbox{Self-built system on the left side and off-the-shelf system on the right side. \capt{The two types of real physical systems offer different dynamics.}
    \label{fig:pendulums}}
    {\includegraphics[width=0.69\linewidth]{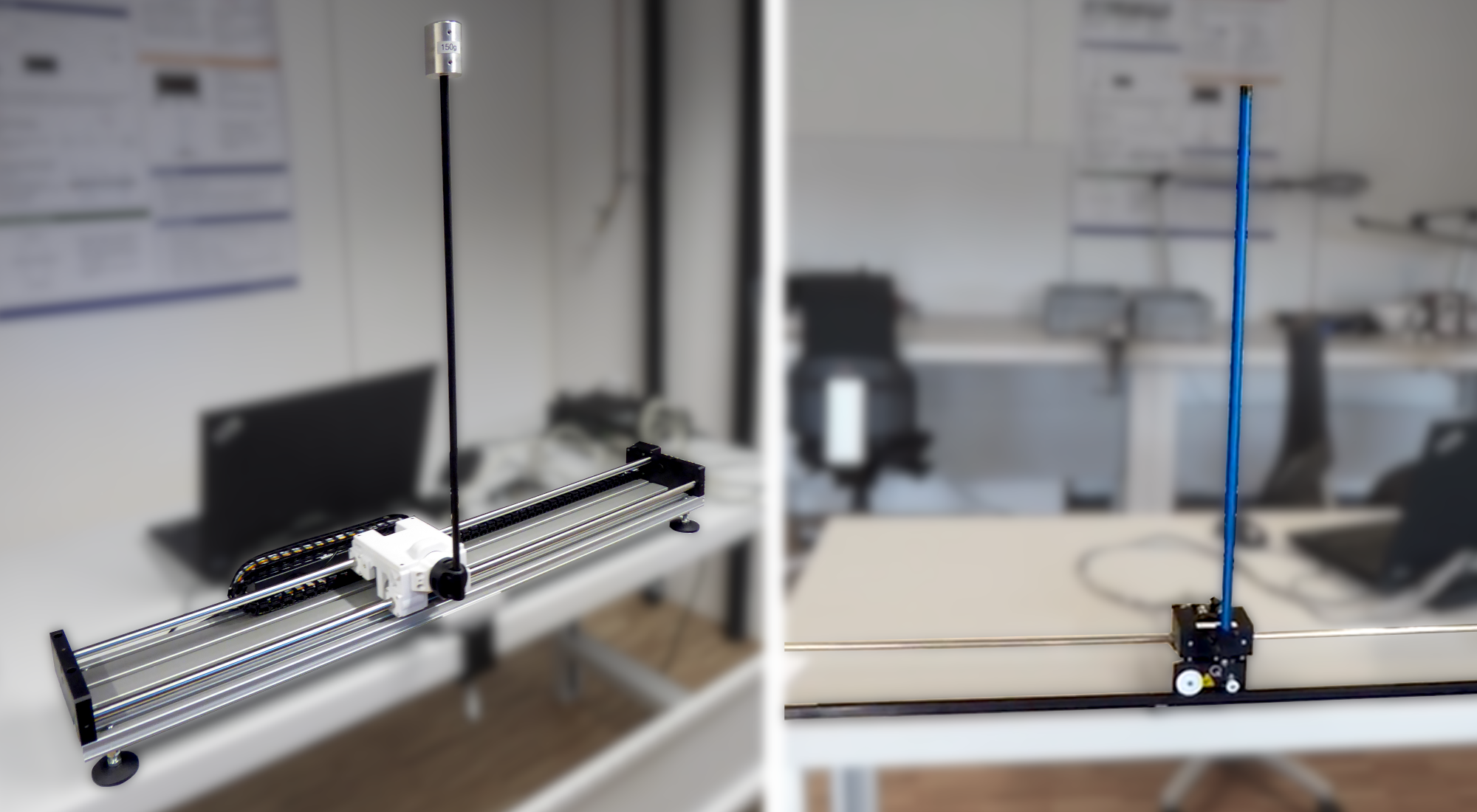}}
    \caption{Wireless cyber-physical testbed with 20 agents. Each agent consists of a wireless node (\dpp) and a real or simulated cart-pole system as physical system. The wireless nodes form a three-hop network.}
    \label{fig:testbed}
\end{figure}

\section{Testbed Experiments}
\label{sec:eval}

This section provides measurements from a 20-agent wireless cyber-physical testbed used to study both the effectiveness of our co-design approach and the interaction of predictive triggering with the communication system.
Evaluating such approaches on real-world testbeds is essential to build trust in the proposed solutions~\cite{lu2015real}.
In our experiments, we find:
\begin{itemize}
    \item The priorities determined by the predictive triggering framework reasonably reflect the need for communication depending on the difference between estimated and actual system states.
    The communication system effectively allocates the available network bandwidth to the agents with the highest need.
    \item Compared to a baseline using periodic control, our approach reduces control costs as defined in~\eqref{eqn:lqr_cost} by up to $\approx \SI{19}{\percent}$, leading to a better synchronization of the 20 agents in our cyber-physical testbed.
\end{itemize}

\subsection{Wireless Cyber-physical Testbed}
\label{sec:testbed}

\paragraph{Characteristics}
We use a wireless cyber-physical testbed with 20 agents that spans an area of about \SI{20}{\meter}~$\times$~\SI{30}{\meter} in our robotic lab.
\figref{fig:testbed_layout} shows the locations of the agents in our testbed.
Each agent consists of a wireless node (see \secref{sec:integration} and \figref{fig:dpp}) and a physical system.

\paragraph{Wireless Network}
We use the \SI{2}{Mbps} BLE mode and set the transmission power to \SI{-8}{\deci\belmilliwatt}.
With these settings, paths between wireless nodes are between 1 and 3 hops in length.
Communication in our testbed is subject to interference from other wireless transmitters operating in the \SI{2.4}{\giga\hertz} band and radiation of electrical components from co-located robotic experiments.

\paragraph{Physical Systems}
We use \emph{cart-pole systems} as physical systems (see \figref{fig:cartpole_system} and \figref{fig:pendulums}) because they exhibit fast dynamics and are widely studied in control theory~\cite{boubaker2012inverted}.
In particular, we aim at stabilizing the pole attached to the cart at a pole angle of $\theta=\SI{0}{\degree}$.
This is known as the \emph{inverted pendulum}, which can be well approximated by an LTI system for small deviations from the unstable equilibrium (\ie $\sin(\theta) \approx \theta$).
The state $x(k)$ of the system consists of the pole angle~$\theta(k)$, the cart position~$s(k)$, and their respective derivatives.
We measure $\theta(k)$ and $s(k)$ using angle sensors, and estimate $\dot{\theta}(k)$ and $\dot{s}(k)$ using finite differences and low-pass filtering.
The control input $u(k)$ is the voltage applied to the motor of the cart, steering the cart's movement speed and direction.

As shown in \figref{fig:testbed_layout}, our testbed features 6 real cart-pole systems connected to the \ap of nodes 2, 3, 5--8, and 14 simulated cart-pole systems by running simulation models on the AP of nodes 1, 4, 9--20.
The simulated systems allow us to scale up the number of agents in our testbed, given the hefty price tag of real cart-pole systems.
In particular, we use 2 off-the-shelf cart-pole systems from Quanser (agent 7 and 8) and 4 cart-pole systems we built in our lab (agent 2, 3, 5, and 6).
As a result, we have a testbed with three types of physical systems with different properties and communication demands.
In general, heterogeneity makes distributed control more challenging compared to homogeneous systems~\cite{lunze2012synchronization}.

\subsection{Scenario and Settings}
\label{sec:scenario}

\paragraph{Control Task}
We consider a scenario where the objective and timescales of the distributed control task are similar to, for example, coordination in drone swarms.
Drones have to locally stabilize their flight and, additionally, exchange information to coordinate their movements in order to keep a certain formation or avoid crashing into each other.
Similarly, all $\nAgent=20$ agents in our testbed have to locally stabilize their poles in an upright position at cart position $s = \SI{0}{\centi\meter}$ (\ie in the middle of their tracks) \emph{and} to synchronize their cart positions.
Synchronizing their cart position requires that all agents know the states of all others.
Stabilization, however, depends only on information every node has locally.
Thus, local update intervals are not bound to the communication delay and can be shorter, making stability less critical.
In our setup, we run the local control loop with an update interval of \SI{10}{\milli\second}.
For synchronization, agents exchange information over the multi-hop wireless network with a communication period of \SI{100}{\milli\second}.

As discussed in \secref{sec:viol_prob}, systems need to exchange their states because the state predictions using the deterministic model will deviate from the real-world over time.
However, in reality, disturbances, such as a gust of wind affecting a subset of drones, can create even more significant deviations.
To include this aspect in our experiments, we let all carts start at $s = \SI{0}{\centi\meter}$ but purposely fix the cart position of one of the simulated systems (agent 11) at $s = \SI{20}{\centi\meter}$ half-way through an experiment.

\begin{figure}[!tb]
    \begin{subfigure}[t]{\linewidth}
        \centering
        \input{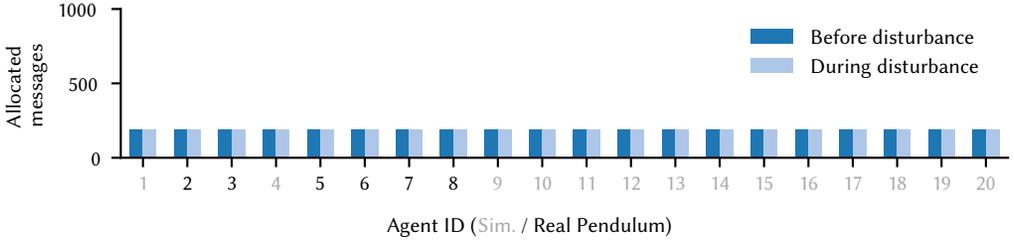}
        \caption{Periodic.
        \capt{
            Each agent is allocated the same number of messages regardless of the disturbance.
        }}
        \label{fig:resources_periodic}
    \end{subfigure}
    \begin{subfigure}[t]{\linewidth}
        \centering
        \input{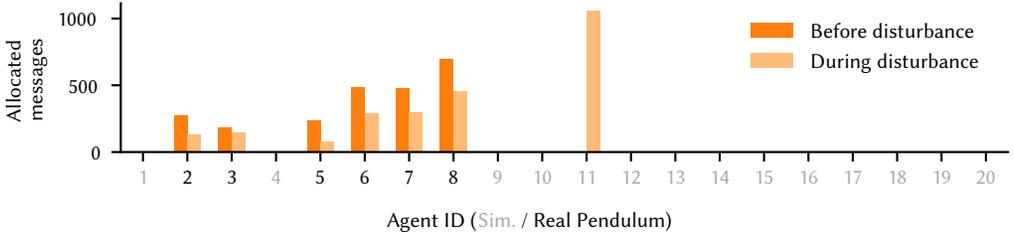}
        \caption{Predictive.
        \capt{
            Messages are allocated according to the agents' needs, correctly addressing the disturbance. %
        }}
        \label{fig:resources_predictive}
    \end{subfigure}
    \caption{
        Total number of control messages allocated to each agent throughout the experiments for the predictive and periodic approach.
        A disturbance is added to agent 11 in the middle of the experiments.
        \capt{
            Compared to the periodic approach, the predictive approach assigns the control messages based on the current state of the agents.
    }}
    \label{fig:resources}
\end{figure}

\paragraph{Comparison Approaches}
In addition to the control task, high-priority application traffic, such as image data recorded by drones in remote surveillance~\cite{Hayat2016}, occupies most of the available wireless communication bandwidth.
Specifically, every communication round, $\nMsgApp = 18$ messages are reserved for the application.
The remaining network bandwidth is used by the agents to exchange the information required for synchronization.
All messages have a payload size of \SI{32}{\byte}.
We compare our approach against \emph{periodic} control.
These two differ in the number of \nMsgCntrl messages and the way these messages are assigned to the agents.
Further details regarding the calculation of \nMsgCntrl is provided in the supplementary material in \secref{sec:app_remaining_bandwdith}.
To allow for a fair comparison between the two approaches, both use Mixer with the enhancements described in \secref{sec:scalability}.
The mechanisms from \secref{sec:triggering_support}, designed to support predictive triggering, are only used in our \emph{predictive} approach.
ETC and STC approaches are not considered here, because it is unclear how they can systematically address overload situations.

\underline{Periodic:} Periodic control introduces no overhead.
Thus, the remaining network bandwidth can be used for $\nMsgCntrl = 3$ messages per communication round in our scenario.
Since there are $\nAgent = 20$ agents that want to communicate each round, the system is overloaded.
The best we can do with periodic control is to allocate the \nMsgCntrl messages in a predetermined round-robin fashion; that is, each agent is allocated a control message every 7th communication round, effectively sharing its state every \SI{700}{\milli\second}.

\underline{Predictive:} This represents the novel predictive triggering approach presented in this paper, where in each communication round the \nMsgCntrl agents with the highest priorities are allowed to share their states.
To efficiently compute priorities on our embedded hardware we approximate~\eqref{eqn:viol_prob_analytic} with Chernoff bounds~\cite[Lem.~2.2]{dasgupta2003elementary}.
With a continuous priority measure, the probability that two nodes transmit the same priority is 0.
In an implementation, however, priorities must be quantized.
We choose $\aggsize = \SI{4}{bit}$ wide priorities, which lets us distinguish \num{16} different values.
If there is a tie, we sort the nodes by their IDs in ascending order.
With the priority aggregate in the packet header (see \figref{fig:comProtocol}) the size of all packets increases, which reduces the network bandwidth available for the control messages.
In our experiments, this reduces \nMsgCntrl to 2 control messages per round with an aggregate overhead of $S=\SI{5}{\byte}$ (see \eqref{eqn:aggregate_size}).
Further, we use a threshold of $P_\delta=0.5$ and choose $\hat{e}_\mathrm{max}=(0.03,0.03,0.1,0.3)^\transp$ as the maximum error to derive the priorities (\cf \secref{sec:viol_prob}).
In case any individual error state grows above the corresponding entry in $\hat{e}_\mathrm{max}$ in terms of its absolute value, we transmit the highest priority.

\paragraph{Controller Design}
For both approaches, we set all $Q_i$ matrices as suggested in~\cite{Quanser2012}, $R=0.1$, and set the first entry of $Q_\mathrm{sync}$ to 10 and all others to 0 to express our desire to synchronize the cart positions.
We adopt the system matrices provided in~\cite{Quanser2012} for the off-the-shelf systems and the ones from~\cite{mastrangelo2019predictive} for the self-built systems.
The respective matrices are used for controller design and estimation.
The simulated systems also use the matrices from~\cite{Quanser2012}, but we add normally distributed noise with a standard deviation of $10^{-4}$.

\subsection{Results}

In the following, we discuss the results for our predictive triggering approach and periodic control.
The results are representative for our concrete scenario and may differ in other cases.
Due to the design of the physical systems used, in particular the limited track length of the carts, the estimation error $\hat{e}$ is restricted to a relatively small value, which favors the periodic control approach.
In scenarios where $\hat{e}$ can grow larger or is potentially unbounded, we expect higher control performance gains with our predictive approach.

\paragraph{Resource Allocation}
In \figref{fig:resources} we show the total number of control messages allocated to each agent throughout one characteristic experiment for each approach.
We begin by discussing the first half of the experiment, before the disturbance, and continue with the second half afterwards.

\underline{Before the disturbance:}
With the \emph{periodic} approach, every agent gets the same share of messages.
Using the \emph{predictive} approach, instead, some agents get to communicate significantly more or less often than others.
This is because the \emph{predictive} approach dynamically allocates messages to agents based on their actual needs.
We see that agents with simulated systems (1, 4, 9--20) rarely get messages allocated ($\le 12$ messages per agent) because their states can be predicted more accurately. %
The states of agents with real systems (2, 3, 5--8) cannot be predicted as accurately as those of the simulated systems, so they need to communicate their states more often to achieve a better synchronization.
This higher need is effectively accounted for by the \emph{predictive} approach, as visible from the many messages allocated to these agents. %

\underline{During the disturbance:}
In the middle of the experiment the position of agent 11 is fixed at $s=\SI{20}{\centi\meter}$, which suddenly leads to very inaccurate state predictions for this agent.
This does not affect the message allocation in the \emph{periodic} approach and every agent still gets the same amount of messages.
The \emph{predictive} approach, on the other hand, adapts to the situation by allocating agent 11 the most messages among all agents ($>1000$ messages) during the second half of the experiment. %
The instantaneous trigger, \ie $P_i^0(k+\mathrm{\horizon}) > P_\delta$ in \eqref{eqn:comm_dec}, reduces the number of control messages by about 1--\SI{2}{\percent}.
More savings are expected when the ratio $\nMsgCntrl / \nAgent$ increases.

\begin{figure}[!tb]
    \centering
    \input{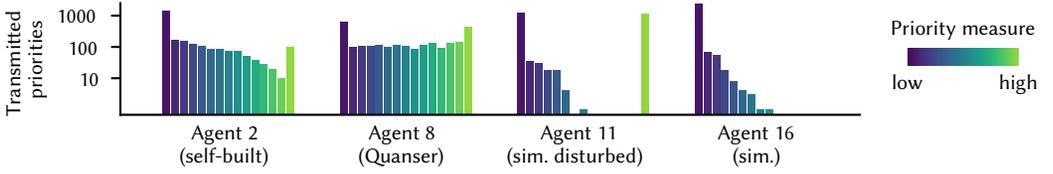}
    \caption{
        Per agent distribution of the transmitted priorities for a selected set of agents using the predictive approach.
        \capt{
            Systems that are more difficult to predict will transmit high priorities more often.
            Simulated systems transmit mostly low priorities, except agent 11 that is disturbed during half of the experiment.
    }}
    \label{fig:priorities}
\end{figure}

\paragraph{Priority Distribution}
In the \emph{predictive} approach, the decision about which agent is allocated a message is based on the communication needs, which is reflected by the transmitted priorities.
\figref{fig:priorities} shows how often each priority was transmitted during the experiment, for a selected set of agents.
The selection includes one representative agent for each type of pendulum (self-built, Quanser, simulated), and additionally the disturbed agent 11.
Each histogram indicates how often the respective agent has transmitted low to high priorities.
Agent 2 and 8 have different types of real physical systems (self-built vs. off-the-shelf).
Although both are cart-pole systems and, thus, have similar dynamics, their individual physical characteristics are different, resulting in different communication needs.
Agent 8 transmits higher priorities more often compared to agent 2, implying that the state of agent 8 is less predictable.
Agents 11 and 16 have simulated systems, so they show very similar communication needs.
However, as soon as we manipulate the cart position of agent 11, its predictions become inaccurate, and it chooses higher priorities more often; in fact, during the disturbance, it chooses the highest priority most often among all agents.
Consequently, the \emph{predictive} approach allocates a message to agent 11 in almost every communication round, helping the other agents to better synchronize the cart positions.
The priority distribution of the different pendulum types is so characteristic across all our experiments that the type of each pendulum can be identified from this alone.
Moreover, these distributions reflect the different communication needs which also shows the heterogeneity of our physical systems.

\begin{figure}[!tb]
    \centering
    \input{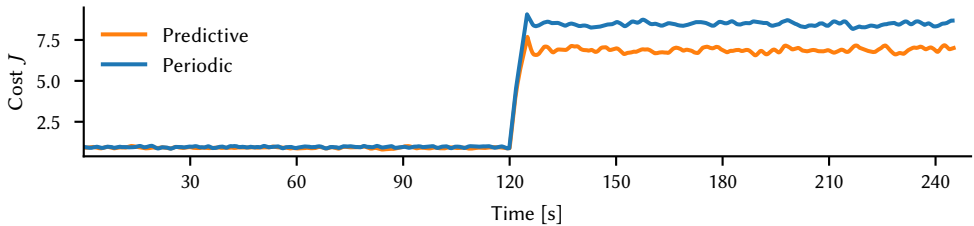}
    \caption{
        Control cost over time, averaged across multiple runs per approach.
        \capt{
            During the disturbance (starting at \SI{120}{\s}), the predictive approach decreases the control cost by around \SI{19}{\percent} due to the adaptive allocation of communication resources.
    }}
    \label{fig:costs}
\end{figure}

\paragraph{Control Performance}
\figref{fig:costs} shows the control performance by evaluating the cost function~\eqref{eqn:lqr_cost} in each time step.
The results are computed using a moving average with a window size of 50 time steps (\SI{5}{\s}), and averaged over 3 different runs per approach.
In the first half of the experiment, both approaches are close together with $\approx \SI{5}{\percent}$ lower cost for the \emph{predictive} approach.
The distance increases in the second half of the experiments, during the disturbance, where the control cost of our approach is $\approx \SI{19}{\percent}$ lower, because it adaptively allocates more messages to the disturbed agent 11, thus enabling a better synchronization.
This improvement is achieved even though the \emph{periodic} approach can transmit one control message more per communication round ($\nMsgCntrl=3$) compared to the \emph{predictive} approach ($\nMsgCntrl=2$).

In our experiments, the network bandwidth reserved for application traffic is relatively high ($\approx\SI{85}{\percent}$).
Conversely, the bandwidth remaining for control messages is small.
As a result, the overhead of our \emph{predictive} approach reduces \nMsgCntrl from 3 (with \emph{periodic} control) to 2, which is \SI{33}{\percent} less.
However, the overhead depends on the concrete experiment settings.
For example, in the same scenario but with $\nMsgApp = 10$ messages, the periodic approach could use $\nMsgCntrl = 11$ control messages and the predictive approach $\nMsgCntrl = 9$ control messages, with an aggregate overhead of $S = \SI{10}{bytes}$.
In this case, the number of control messages is only reduced by \SI{18}{\percent}.

\section{Conclusions}
\label{sec:ending}

We present a novel co-design of distributed control and wireless communication that successfully handles overload, \ie settings where we have more agents that want to transmit information than the communication system can support.
By deriving a priority measure in advance and efficiently distributing this measure among all agents, we can assign the few communication resources to the agents with the greatest need for communication.
We provide formal stability guarantees for the entire cyber-physical system, including control and communication systems.
Furthermore, we evaluate our approach on a real cyber-physical testbed, demonstrating support for synchronization of 20 cart-pole systems, at update intervals of \SI{100}{\milli\second} over a 3-hop wireless network using low-power embedded hardware.

\begin{acks}
	We thank Harsoveet Singh, Felix Grimminger, and Joel Bessekon Akpo for their help with the cyber-physical testbed, and Alexander Gräfe for important discussions and careful proofreading. We also thank the TEC group at ETH Zurich for the publicly available design of the DPP, and Kai Geißdörfer for designing the new DPP used in this work. This work was supported in part by the \grantsponsor{dfg}{German Research Foundation}{}~(DFG) through SPP 1914 (grants \grantnum{dfg}{ZI 1635/1-1}, \grantnum{dfg}{TR 1433/1-1}, and \grantnum{dfg}{TR 1433/1-2}), and the Emmy Noether project NextIoT (grant \grantnum{dfg}{ZI 1635/2-1}), the Cyber Valley Initiative, the Max Planck Society, and the Center for Advancing Electronics Dresden (cfaed).
\end{acks}

\bibliographystyle{ACM-Reference-Format}
\bibliography{ref}

\newpage

\appendix
\section{Proof of Theorem~\ref{thm:stability}}

Here, we provide the proof of Theorem~\ref{thm:stability}.
We start with some preliminaries that will be needed to ensure that the estimation errors $\hat{e}_i(k)$ and $\hat{e}_{ij}$ are well-behaved, what lets us conclude MSB of the overall system.

\subsection{Preliminaries}

To study the behavior of the estimation error, we employ the concept of $f$-ergodicity.
Generally, a stochastic process as~\eqref{eqn:gen_sys_lin_dyn} is said to be ergodic if its sample average and time average coincide.
The notion of $f$-ergodicity is stronger a stronger notion and is used in the context of Markov chains.
Intuitively, if a process is $f$-ergodic, the Markov chain is stationary and the process itself converges to an invariant finite-variance measure over the entire state-space.
More formally:
\begin{defi}[{\cite[Ch.~10]{meyn2012markov}}]
Let the Markov chain $\Phi=(\Phi(0),$ $\Phi(1),$ $\ldots)$ evolve in the state-space $\mathcal{X}$, which is equipped with some known $\sigma$-algebra $\mathcal{B}(\mathcal{X})$.
The Markov chain $\Phi$ is said to be positive Harris recurrent (PHR) if
\begin{enumerate}
	\item a non-trivial measure $\nu(B)>0$ exists for a set $B\in\mathcal{B}$ such that for all $\Phi(0)\in\mathcal{X}$, $P(\Phi(k)\in B, k<\infty)=1$ holds.
	\item $\Phi$ admits a unique invariant probability measure.
\end{enumerate}
\end{defi}
\begin{defi}[{\cite[Ch.~14]{meyn2012markov}}]
Let $f\ge 1$ be a real-valued function in $\R^n$.
A Markov chain $\Phi$ is said to be $f$-ergodic, if
\begin{enumerate}
	\item $\Phi$ is PHR with unique invariant measure $\pi$,
	\item the expectation $\pi(f)=\int f(\Phi(k))\pi(\diff\Phi(k))$ is finite,
	\item $\lim_{k\to\infty}\lVert P^k(\Phi(0),\cdot)-\pi\rVert_f = 0$ for every initial value $\Phi(0)\in\mathcal{X}$, where $\norm{\nu}_f=\sup_{\abs{g}\le f}\abs{\nu(g)}$.
\end{enumerate}
\end{defi}
We further define a drift function:
\begin{defi}[\cite{mamduhi2017error}]
Let $V:\R^n\to\R_{\ge 0}$, and $\Phi$ be a Markov chain.
For any measurable function $V$, the drift $\Delta V(\cdot)$ is
\begin{align}
\label{eqn:drift}
\Delta V(\Phi(k)) \coloneqq \E[V(\Phi(k+1))\mid\Phi(k)] - V(\Phi(k)),
\end{align}
with $\Phi(k)\in\R^n$.
\end{defi}
We can then establish f-ergodicity.
\begin{theo}[{$f$-Norm Ergodic Theorem~\cite[Ch.~15]{meyn2012markov}}]
\label{thm:f-ergodicity}
Suppose that the Markov chain $\Phi$ is $\psi$-irreducible and aperiodic and let $f(\Phi(k))\ge 1$ a real-valued function in $\R^n$.
If a small set $\mathcal{D}$ and a non-negative real-valued function $V$ exist such that $\Delta V(\Phi(k)) \le -f(\Phi(k))$, for every $\Phi(k)\in\R^n\setminus\mathcal{D}$, and $\Delta V<\infty$ for $\Phi(k)\in\mathcal{D}$, the Markov chain $\Phi$ is $f$-ergodic.
\end{theo}
The Markov chain defined through the dynamical system~\eqref{eqn:gen_sys_lin_dyn} is both aperiodic and $\psi$-irreducible.
This is the case since the noise distribution $v_i(k)$ is assumed to be absolutely continuous with an everywhere-positive density function.
Thus, every subset of the state space $\mathcal{X}$ is reachable within one time-step.
Here, $\psi$ is a non-trivial measure on $\R^n$.
Further, all compact subsets of an LTI system are small sets~\cite[Ch.~5]{meyn2012markov}.

\subsection{Stochastic Stability}

Equipped with the concepts from the previous section, we now present the stability proof.
For this, we first show that the Markov chain induced by the estimation errors $\hat{e}_i(k)$ and $\hat{e}_{ij}$ is $f$-ergodic and then show that this implies MSB of the overall system.

For analyzing the ergodicity of the error Markov chain we start with defining the function $V$ as the sum of all squared Mahalanobis distances from the equilibrium, \ie
\begin{align}
\label{eqn:sum_of_distances}
V(\hat{e}(k)) = \sum\limits_{i=0}^\nAgent \mah{\hat{e}_i(k)}^2 + \sum\limits_{i,j=0, i\neq j}^N \mah{\hat{e}_{ij}(k)}^2.
\end{align}
Ideally, we would now seek to guarantee $\Delta V(\hat{e}(k))$ to be negative at every time step to invoke $f$-ergodicity.
However, for an event-triggered strategy and a lossy communication channel, this is not possible.
We may have rounds, in which none of the agents seeks to communicate as all their errors are low enough.
Also, we may have agents that want to communicate but messages being lost.
In both cases, the drift may become positive.
Thus, we modify the drift definition from~\eqref{eqn:drift},
\begin{align}
\label{eqn:drift_modified}
\Delta V(\hat{e}(k),K) = \E[V(\hat{e}(k+K))\mid\hat{e}(k)]-V(\hat{e}(k)).
\end{align}
As discussed in~\cite[Ch.~19]{meyn2012markov}, Theorem~\ref{thm:f-ergodicity} can also be used to show $f$-ergodicity for such modified drift definitions.
Here, we choose $K=\text{ceil}(2 \nAgent / \nMsgCntrl)$, where ceil is a function rounding its argument to the next higher integer.
The incentive of this is to have an interval that allows each agent to communicate once.
Since agents announce their communication needs in advance and only learn about whether they received a slot in the round they receive it, it might happen that an agent is awarded two slots in a row.
To also allow for communication of all agents in such cases, we take two times the minimum interval.

In the absence of message loss, we have $\hat{e}_i(k)=\hat{e}_{ij}(k)$ for all $k$ and $j$.
In the following, we, thus, first derive bounds for $\hat{e}_i(k)$ and then show how we deal with errors due to message loss.
Note that the estimation error $\hat{e}_i(k+K)$ can be written as a function of $\hat{e}_i(k)$,
\begin{align}
\label{eqn:err_expl_form}
\begin{split}
\hat{e}_i(k+K) &= \tilde{A}^{K-k}\hat{e}_i(k)\prod_{n=k}^K(1-\kappa_i(n)\phi_i(n))\\
&+ \sum_{n=k}^{K-1}\tilde{A}^{K-k-1}v_i(K-n)\prod_{r=n+2}^K(1-\kappa_i(r)\phi_i(r)),
\end{split}
\end{align}
where $\tilde{A} = A + BF_{ii}$.

We now seek to upper bound the expected squared Mahalanobis distance of $\hat{e}_i(k)$.
Generally, such bounds can be derived as follows:
\begin{lem}
\label{lem:upperBound}
We can upper bound $\E[\mah{\hat{e}_i(k+1)}^2\mid \hat{e}_i(k)]$ in case of no communication by
\begin{align}
\label{eqn:upperBound}
\E[\mah{\hat{e}_i(k+1)}^2\mid \hat{e}_i(k)]\le\norm{\tilde{A}}_2^2\mah{\hat{e}_i(k)}^2 + \Tr(I_n).
\end{align}
\end{lem}
\begin{proof}
Following~\eqref{eqn:err_expl_form}, we have\footnote{Note that here we ignore the additional information that if the system does not receive a communication slot, its error must have been lower than that of the others. This ``negative information''~\cite{sijs2013event} could be used to arrive at sharper bounds.}
\begin{align}
\E[\mah{\hat{e}_i(k+1)}^2\mid \hat{e}_i(k)] = \E[\mah{\tilde{A}\hat{e}_i(k)+v_i(k)}^2\mid \hat{e}_i(k)].
\end{align}
Leveraging that the noise has zero mean and using the Cauchy-Schwarz inequality, we arrive at
\begin{align}
\begin{split}
&\E[\mah{\tilde{A}\hat{e}_i(k)+v_i(k)}^2\mid \hat{e}_i(k)] \\
&= \E[\mah{\tilde{A}\hat{e}_i(k)}^2\mid \hat{e}_i(k)]+\E[\mah{v_i(k)}^2]\\
&\le \norm{\tilde{A}}_2^2\E[\mah{\hat{e}_i(k)}^2\mid \hat{e}_i(k)] + \Tr(\Sigma_\mathrm{v_i}^{-1}\Sigma_\mathrm{v_i})\\
&=\norm{\tilde{A}}_2^2\mah{\hat{e}_i(k)}^2 + \Tr(I_n).
\qedhere
\end{split}
\end{align}
\end{proof}

We now prove $f$-ergodicity of the Markov chain induced by $\hat{e}_i$.
\begin{theo}
\label{thm:ergodicity_error}
Consider the setting from Theorem~\ref{thm:stability}.
Then, for any $\delta\in\R_{\ge 0}$, the Markov chain $\hat{\tilde{e}}_i(k)=[\hat{e}_1(k),\ldots,e_\mathrm{\nAgent}(k)]^\transp$ is $f$-ergodic.
\end{theo}
\begin{proof}
We let the dynamical system evolve over the time interval $[k,k+K]$ and study all possible outcomes.
For this, we partition the agents in mutual exclusive subsets:
\begin{enumerate}
	\item[$c_1$] The agent has or has not transmitted over $[k,k+K]$ and $\mah{\hat{e}_i(k+K-1)}^2\le\delta$;
	\item[$c_2$] The agent has transmitted successfully at least once over $[k,k+K]$ and $\mah{\hat{e}_i(k+K-1)}^2>\delta$;
	\item[$c_3$] The agent has never been assigned a slot over $[k,k+K]$ and $\mah{\hat{e}_i(k+K-1)}^2>\delta$.
\end{enumerate}
We now derive upper bounds for all cases.
For case $c_1$, we know that the estimation error at the last time instant was smaller than the threshold $\delta$.
Thus, we can use this and Lemma~\ref{lem:upperBound} to arrive at
\begin{align}
\label{eqn:proof_case1}
\E[\mah{\hat{e}_i(k+K)}^2\mid\hat{e}_i(k)]\le \delta\norm{\tilde{A}}_2^2 + \Tr(I_n).
\end{align}
For case $c_2$, assume that the last transmission instant was at $k=r$.
Thus, we know that at $k=r$ the error was reset to $0$ and obtain
\begin{align}
\label{eqn:proof_case2}
\E[\mah{\hat{e}_i(k+K)}^2\mid\hat{e}_i(k)]\le \sum\limits_{s=r}^K \Tr(I_n)\norm{\tilde{A}^{K-s}}_2^2.
\end{align}

For case 3, we consider two subcases:
\begin{enumerate}
	\item[$c_{3\mathrm{a}}$] Over $[k,k+K]$, for at least one $k'$ we had $\mah{\hat{e}_i(k')}^2\le\delta$;
	\item[$c_{3\mathrm{b}}$] Over $[k,k+K]$, we had $\mah{\hat{e}_i(k')}^2>\delta$ for all $k'$.
\end{enumerate}
For case $c_{3\mathrm{a}}$, assume that $k=r$ was the last time we had $\mah{\hat{e}_i(k')}^2\le\delta$.
Then, we have, again using Lemma~\ref{lem:upperBound},
\begin{align}
\label{eqn:proof_case3a}
\E[\mah{\hat{e}_i(k+K)}^2\mid\hat{e}_i(k)]\le \delta\norm{\tilde{A}^{K-r}}_2^2 + \sum\limits_{s=r}^{K-1} \Tr(I_n)\norm{\tilde{A}^{K-s-1}}_2^2.
\end{align}

For case $c_{3\mathrm{b}}$, we have agents whose error was above the threshold $\delta$ in every time-step, but never got a slot.
As the scheduling rule always awards resources to the agents with the highest error, we can upper bound the error of agents in $c_{3\mathrm{b}}$ with the worst-case error of the agents in $c_2$,
\begin{align}
\label{eqn:proof_case3b}
\E[\mah{\hat{e}_i(k+K)}^2\mid\hat{e}_i(k)]\le\max_{i\in c_2}\E[\mah{\hat{e}_i(k+K)}^2\mid\hat{e}_i(k)].
\end{align}
The drift can now be upper bounded as
\begin{align}
\label{eqn:proof_drift}
\Delta V(\hat{\tilde{e}}_i(k),K) &\le \sum\limits_{i\in c_1,c_2,c_3}\E[\mah{\hat{e}_i(k+K)}^2\mid\hat{e}_i(k)]-V(\hat{e}(k))\nonumber\\
&\le\sum\limits_{i\in c_1}\left(\delta\norm{\tilde{A}}_2^2+\Tr(I_n)\right)
+\sum\limits_{i\in c_2}\sum\limits_{n=r_i}^K\Tr(I_n)\norm{\tilde{A}^{K-n}}_2^2\nonumber\\
&+ \sum\limits_{i\in c_{3\mathrm{a}}}\left( \delta\norm{\tilde{A}^{K-r}}_2^2 + \sum\limits_{n=r}^{K-1} \Tr(I_n)\norm{\tilde{A}^{K-n-1}}_2^2\right)\nonumber\\
&+\sum\limits_{i\in c_{3\mathrm{b}}}\max_{j\in c_2}\left(\sum\limits_{n=r_j}^K\Tr(I_n)\norm{\tilde{A}^{K-n}}_2^2\right)-V(\hat{e}(k))\nonumber\\
&= \zeta - V(\hat{\tilde{e}}_i(k)),
\end{align}
where $\zeta$ represents all bounded terms stemming from cases $c_1$, $c_2$, and $c_3$.
We can then define $f(\hat{\tilde{e}}(k))=1+\epsilon V(\hat{\tilde{e}}(k))$, with $\epsilon\in(0,1)$.
Since $V(\hat{\tilde{e}}(k))\ge0$, it follows that $f(\hat{\tilde{e}}(k))\ge1$.
As $v(\hat{\tilde{e}}(k))$ grows with $\hat{\tilde{e}}(k))$ while $\zeta$ is constant, we can further find a small set $\mathcal{D}$ and an $\epsilon$ such that $\Delta V(\hat{e}(k),K)\le -f$, which proves $f$-ergodicity according to Theorem~\ref{thm:f-ergodicity} in the absence of message loss.
\end{proof}

With this, we have shown $f$-ergodicity of $\hat{\tilde{e}}_i(k)$.
However, there might, due to message loss, still be a divergence between $\hat{e}_i(k)$ and $\hat{e}_{ij}(k)$, \ie the estimates that other agents have about agent $i$'s state.
\begin{cor}
\label{cor:stab_msg_loss}
Consider the same setting as in Theorem~\ref{thm:ergodicity_error}.
Then, the Markov chain $\hat{e}_=[\hat{e}_{1}(k),\ldots,$ $\hat{e}_{1\mathrm{\nAgent}},\hat{e}_{21},\ldots,e_{\mathrm{\nAgent}}(k)]^\transp$ is $f$-ergodic under \iid message loss.
\end{cor}
\begin{proof}
First, assume that only messages of agent~$i$ might be lost, but not those of other agents.
Then, we can account for lost messages, following the strategy of~\cite{mamduhi2015robust}, by extending the horizon until the next successful message transmission of agent $i$.
That is, we consider $k\to k+K+m^*$ instead of $k\to k+K$.
We then have
\begin{align}
\label{eqn:msg_loss_mstar}
    \E[\mah{\hat{e}_{ij}(k+K+m^*)}^2\mid \hat{e}_{ij}(k)] \le \Tr(I_n).
\end{align}
Thus, for any finite $m^*$, the estimation error is bounded.
Now, consider the case $m\to\infty$ without a single message being transmitted successfully.
Due to the \iid property of message loss, the probability of losing $m_\mathrm{c}$ messages in a row is $(1-p_\mathrm{c})^{m_\mathrm{c}}$, where $p_\mathrm{c}$ is the probability of successfully transmitting a message.
Thus, we have
\begin{align}
\label{eqn:msg_loss_bound}
&\E[\mah{\hat{e}_{ij}(k+K+m)}^2\mid \hat{e}_{ij}(k))] \le\\
&(1-p_\mathrm{c})^{m_\mathrm{c}} \left(\mah{\hat{e}_{ij}(k)}^2\norm{\tilde{A}^{K+m}}_2^2 + \sum\limits_{n=k}^{K+m-1} \Tr(I_n)\norm{\tilde{A}^{K-n-1}}_2^2\right)\nonumber.
\end{align}
The variable $m_\mathrm{c}$ here denotes the amount of times agent $i$ tried to transmit its state in the interval $k\to k+K+m$.
Agent $i$ competes for a slot, if $\mah{\hat{e}_i(k)}>\delta$ for some $k$ and if it gets a slot the error is reset.
Thus, the time at which $\hat{e}_i(k)$ exceeds the threshold $\delta$ can be interpreted as a stopping time, which has finite expected value~\cite[Sec.~2]{grebenkov2014first}.
Therefore, agent $i$ will infinitely often compete for slots as $m\to\infty$.
The same holds for all other agents and since we assume homogeneous agents, each will be assigned the same amount of slots as $m\to\infty$ due to the law of large numbers.
In conclusion, we have $m_\mathrm{c}\to\infty$ as $m\to\infty$.
As we further have $\norm{\tilde{A}}_2<1$ and $p_\mathrm{c}>0$ (\cf the assumptions stated at the beginning of \secref{sec:stab_analysis}),~\eqref{eqn:msg_loss_bound} goes to zero as $m^*$ goes to infinity.

In practice, also other agents may lose messages.
Then, the estimation error $\hat{e}_{ij}$ might grow not only because agent~$j$ is missing information about agent~$i$, but also because their information about other agents states might differ.
Since the control input depends also on these other states, this might lead to diverging predictions.
In this case, we need to add another term
\begin{align}
\label{eqn:msg_loss_l}
\sum_{\ell\in\Omega_i\setminus\{i,j\}}\mah{\hat{e}_{i\ell}(k+K+m^*-1)-\hat{e}_{j\ell}(k+K+m^*-1)}^2\norm{\tilde{A}}_2^2
\end{align}
to~\eqref{eqn:msg_loss_mstar}.
Also these terms can be bounded.
In the beginning, we assume agents~$i$ and~$j$ to start with the same initial guess about agent~$\ell$, \ie~\eqref{eqn:msg_loss_l} is zero.
We can now, similar as above, extend the horizon such that at $k+K+m^*$,~$i$ successfully transmitted a message to~$j$ and at some $k'$ within $k\to k+K+m^*$ both~$i$ and~$j$ successfully received an update from~$\ell$.
If then~$i$ receives another update within at some $k''>k'$ while~$j$ did not, the error is bounded by
\begin{align}
\label{eqn:msg_loss_l_div}
\E[\mah{\hat{e}_{i\ell}(k+K+m^*)-\hat{e}_{j\ell}(k+K+m^*)}^2\mid\hat{e}_{ij}(k)]\le \sum_{s=k'}^{K+m^*}\Tr(I_n)\norm{\tilde{A}^{K+m^*-s}} + \sum_{s=k''}^{K+m^*}\Tr(I_n)\norm{\tilde{A}^{K+m^*-s}}.
\end{align}
The only possibility for the error to grow without bounds is for agent~$i$ to receive an update of~$\ell$ at some point while~$j$ never receives one.
For the same arguments as above, the probability of that happening converges to zero faster than the error grows as $m$ goes to infinity.

Thus, also in the case of message loss, the Markov chain is $f$-ergodic.
\end{proof}

\begin{remark}
\label{rem:stab}
Under additional assumptions, the results can be extended to systems for which $\norm{\tilde{A}}_2 > 1$.
In that case, it needs to be ensured that the probability of successfully delivering a message grows faster than the estimation error.
\end{remark}

With this, we can now prove stability of the overall system.
For MSB following Definition~\ref{def:stability}, we need $\E[\norm{x(k)}_2^2]$ to be bounded.
The individual system without estimation error is exponentially stable, \ie its second moment would go to 0 at an exponential rate.
The second moment of the estimation error can be upper bounded, independent of the state of the system, and also the noise variance is constant and independent of the system's state.
That is, we essentially have an exponentially stable system with a constant disturbance.
Thus, also the second moment of the system state is bounded~\cite[Ch.~7.6d]{callier2012linear}.

\section{Usage of the Network Bandwidth for Control}
\label{sec:app_remaining_bandwdith}
The network bandwidth available for control traffic (\nMsgCntrl) depends on the communication demand of the application (\nMsgApp) and other scenario parameters such as the required update interval.
In our experiments in \secref{sec:scenario} the communication period is \SI{100}{\milli\second}, and in each round $\nMsgApp = 18$ application messages are sent.
In between two communication rounds, each agent does computations based on the received data, and provides the current priority and state information for the next communication round.
For these calculations and some additional buffer time we need to reserve \SI{24}{\milli\second}, thus, \SI{76}{\milli\second} remain for the communication system to exchange data.
Based on our BLE experiments in \secref{sec:scalability}, we can derive that Mixer needs on average about \num{9.5} slots per message.
In the \SI{2}{Mbps} BLE mode it takes \SI{4}{\micro\second} to transmit one byte.
With respect to the packet size and other unavoidable delays (see \cite{herrmann2018mixer}), \eg switching the radio from receive to transmit mode and vice versa, the resulting duration of one Mixer slot is \SI{380}{\micro\second}.
We can now calculate \nMsgCntrl, first for the \emph{periodic} approach and then for our \emph{predictive} approach.

\subsection{Periodic}
The \emph{periodic} approach has no additional overhead and \nMsgCntrl can be straightforwardly calculated with
\begin{equation}
	\label{eqn:mc_periodic}
	\mathrm{communication\_time} \ge (\nMsgApp + \nMsgCntrl) * \mathrm{slots\_per\_message} * \mathrm{slot\_time}.
\end{equation}
Using the specific values of our experiment, we have
\begin{equation*}
	\SI{76000}{\micro\second} \ge (18 + \nMsgCntrl) * 9.5 * \SI{380}{\micro\second},
\end{equation*}
resulting in a value of $\nMsgCntrl = 3$ control messages per communication round.

\subsection{Predictive}
The relationships are a bit more complex in our \emph{predictive} approach, because \nMsgCntrl and the duration of a slot in Mixer are interrelated.
The $\mathrm{slot\_time}$ term in \eqref{eqn:mc_periodic} additionally depends on the aggregate overhead $S$ as calculated in \eqref{eqn:aggregate_size}.
This requires \nMsgCntrl to be calculated iteratively by choosing a value for \nMsgCntrl, then calculating $S$, and finally checking if this is below the maximum communication time.
In our experiments, with $\nAgent = 20$, $\aggsize = 4$, and $\nMsgCntrl = 2$ we get $S = 5$ additional bytes per packet.
As a result, $\mathrm{slot\_time}$ increases from \SI{380}{\micro\second} to \SI{400}{\micro\second} and we get
\begin{equation*}
	\SI{76000}{\micro\second} \ge (18 + 2) * 9.5 * \SI{400}{\micro\second}.
\end{equation*}

\end{document}